\documentclass[10pt,aps,pra,nofootinbib,twocolumn,floatfix,superscriptaddress]{revtex4-2}

\usepackage[T1]{fontenc} 
\usepackage[utf8]{inputenc}
\usepackage[english]{babel}
\usepackage{times}
\usepackage{float}

\usepackage{amsmath,amsthm,amssymb}
\usepackage{graphicx,epsfig,tikz,xcolor}
\usepackage{pifont,bbding}
\usepackage{tikz}
\usetikzlibrary{arrows.meta,positioning,calc}
\usepackage{xurl}
\usepackage[
hyperindex,
breaklinks,
unicode, 
linkcolor=blue,
colorlinks=true,
citecolor=blue,
urlcolor=blue
]{hyperref}

\usepackage[capitalize]{cleveref}
\usepackage{nameref}
\usepackage{csquotes}
\usepackage{indentfirst}

\usepackage{amsthm}
\usepackage{mathtools}
\usepackage{physics}
\usepackage{multirow} 
\usepackage{braket}

\theoremstyle{plain}
\newtheorem{theorem}{Theorem}
\newtheorem{lemma}[theorem]{Lemma}
\newtheorem{proposition}[theorem]{Proposition}
\newtheorem{corollary}[theorem]{Corollary}

\theoremstyle{definition}
\newtheorem{definition}{Definition}

\newcommand{\ten}{\otimes}

\newcommand{\proj}[1]{\ket{#1}\!\!\bra{#1}}
\newcommand{\expA}[1]{\langle#1\rangle}
\renewcommand{\tr}{\mathrm{tr}}
\newcommand{\C}{\mathbb{C}}

\renewcommand{\l}{\left}
\renewcommand{\r}{\right}

\renewcommand{\a}{\alpha}

\newcommand{\HS}{\mathcal{H}}

\newcommand{\id}{\mathbb{I}}

\newcommand{\vecP}{\vec{P}}
\newcommand{\A}{{\rm A}}
\newcommand{\B}{{\rm B}}
\newcommand{\tC}{{\rm C}}
\newcommand{\AB}{{\rm AB}}
\newcommand{\BA}{{\rm BA}}
\newcommand{\PA}{P_\A}
\newcommand{\PB}{P_\B}
\newcommand{\PAB}{P_\AB}
\renewcommand{\L}{\mathcal{L}}
\newcommand{\Q}{\mathcal{Q}}
\newcommand{\ICHSH}{I_\text{CHSH}}

\DeclareFontEncoding{LS1}{}{}
\DeclareFontSubstitution{LS1}{stix}{m}{n}
\DeclareSymbolFont{stixletters}{LS1}{stix}{m}{it}
\DeclareMathAccent{\lrvec}{\mathord}{stixletters}{"95}

\newcommand{\Qstr}{\mathfrak{Q}}
\newcommand{\SQS}{\lrvec{\Qstr}}
\newcommand{\mMQS}{\Qstr^\star_{\Dm}}
\newcommand{\mMSQS}{\lrvec{\Qstr}^\star_{\Dm}}

\newcommand{\Psym}{\vecP_{\scalebox{0.6}{$\leftrightarrow$}}}
\newcommand{\BIsym}{\vec{\beta}_{\scalebox{0.6}{$\leftrightarrow$}}}
\newcommand{\Dm}{d_\text{min}}

\newcommand{\intv}[1]{[#1]}

\newcommand{\xmark}{\ding{55}}
\newcommand{\cmark}{\ding{51}}

\usepackage{CJKutf8} 
\usepackage{ucs} 
\usepackage[encapsulated]{CJK}

\begin{document}
\title{Trading symmetry for Hilbert-space dimension in Bell-inequality violation}

\author{Hsin-Yu Hsu}
\affiliation{Department of Physics and Center for Quantum Frontiers of Research \& Technology (QFort), National Cheng Kung University, Tainan 701, Taiwan}

\author{Gelo Noel M. Tabia}
\affiliation{Hon Hai (Foxconn) Research Institute, Taipei, Taiwan}
\affiliation{Department of Physics and Center for Quantum Frontiers of Research \& Technology (QFort), National Cheng Kung University, Tainan 701, Taiwan}
\affiliation{Physics Division, National Center for Theoretical Sciences, Taipei 106319, Taiwan}

\author{Kai-Siang Chen}
\affiliation{Department of Physics and Center for Quantum Frontiers of Research \& Technology (QFort), National Cheng Kung University, Tainan 701, Taiwan}

\author{Mu-En Liu}
\affiliation{Universit\'e Paris-Saclay, CEA, CNRS, Institut de physique th\'eorique, 91191, Gif-sur-Yvette, France}
\affiliation{Department of Physics and Center for Quantum Frontiers of Research \& Technology (QFort), National Cheng Kung University, Tainan 701, Taiwan}

\author{Tam\'as V\'ertesi}
\affiliation{HUN-REN Institute for Nuclear Research, P.O. Box 51, H-4001 Debrecen, Hungary}

\author{Nicolas Brunner}
\affiliation{Department of Applied Physics, University of Geneva, 1211 Geneva 4, Switzerland}

\author{Yeong-Cherng Liang}
\email{ycliang@mail.ncku.edu.tw}
\affiliation{Department of Physics and Center for Quantum Frontiers of Research \& Technology (QFort), National Cheng Kung University, Tainan 701, Taiwan}
\affiliation{Physics Division, National Center for Theoretical Sciences, Taipei 106319, Taiwan}
\affiliation{Perimeter Institute for Theoretical Physics, Waterloo, Ontario, Canada, N2L 2Y5}

\begin{abstract}
In quantum information, asymmetry, i.e., the lack of symmetry, is a resource allowing one to accomplish certain tasks that are otherwise impossible. Similarly, in a Bell test using any given Bell inequality, the maximum violation achievable using quantum strategies respecting or disregarding a certain symmetry can be different.  In this work, we focus on the symmetry involved in the exchange of parties and explore when we have to trade this symmetry for a lower-dimensional quantum strategy in achieving the maximal violation of given Bell inequalities. For the family of symmetric Collins-Gisin-Linden-Massar-Popescu inequalities, we provide evidence showing that there is no such trade-off.  However, for several other Bell inequalities with a small number of dichotomic measurement settings, we show that symmetric quantum strategies in the minimal Hilbert space dimension can only lead to a suboptimal Bell violation. In other words, there exist symmetric Bell inequalities that can {\em only} be maximally violated by asymmetric quantum strategies of minimal dimension. In contrast, one can also find examples of asymmetric Bell inequalities that are maximally violated by symmetric correlations. The implications of these findings on the geometry of the set of quantum correlations and the possibility of performing self-testing therefrom are briefly discussed.
\end{abstract}

\date{\today}
\maketitle

\section{Introduction}
    
Quantum nonlocality, i.e., the quantum violation of a Bell inequality~\cite{Bell64}, manifests that quantum theory is incompatible with the notion of local causality~\cite{Bell04}. In particular, {\em no} local-hidden-variable (LHV) theory can reproduce all quantum {\em correlations} (measurement statistics) appearing in a Bell experiment. Apart from its foundational significance, the possibility of device-independent~\cite{Scarani12,Bell-RMP} (DI) quantum information processing (QIP) also arises as an important byproduct of investigating the general phenomenon of Bell nonlocality~\cite{Bell-RMP}. Notable examples of DIQIP protocols include quantum key distribution \cite{Ekert91,Mayers04, Barrett05, Vazirani14}, randomness expansion \cite{Pironio10,Colbeck11}, and various possibilities of black-box certification (see, e.g.,~\cite{Bell-RMP,Bancal15,Liang15,SLChen16,Sekatski2018,Wagner2020,Tavakoli2021,Chen2021robustselftestingof}).

Although the discussion of Bell nonlocality often centers around the lowest-dimensional qubit systems, the Bell violation of higher-dimensional (HD) quantum states has also been explored both theoretically~\cite{KGZ+00,DKZ01,CGLMP,ADG+02,BKP06,LCL07,Liang09,Lim2010,Vertesi2010PRL,SAT+17,KST+19,TBYL22} and experimentally~\cite{TAZ+PRL04,Dada2011Nphysics,Schwarz16,LLY+16}. In fact, it has long been recognized that HD quantum systems can lead to a stronger violation, and hence better resistance to (white) noise~\cite{KGZ+00,DKZ01,CGLMP} and losses~\cite{Vertesi2010PRL}. The local Hilbert space dimensions (HSDs) of the shared state, which reflect the complexity of the underlying degrees of freedom, thus serve as a {\em resource} for demonstrating Bell nonlocality.

However, for any Bell inequality, there is usually a finite-dimensional quantum strategy (QS)---consisting of the state shared by distant observers and their choice of local measurements---that can attain the maximal violation allowed in quantum theory. For example, it suffices~\cite{Masanes06} to consider qubits in maximizing the violation of any Bell inequality involving only two binary-outcome local measurements. The general problem of determining the {\em minimal} HSD required, nevertheless, is highly nontrivial, see, e.g.,~\cite{Pal09QB,PV10,PAC+2025}.

In the case when a Bell inequality comes with a certain symmetry, such as being party-permutation-invariant (PPI)~\cite{Bancal_2010, Bancal12,Liang15,FT17,Aloy2024}, the task of finding its maximal quantum violation can be simplified~\cite{Moroder13} by considering only quantum correlations, and hence QSs that respect the same symmetry. In particular, when combined with a semidefinite programming (SDP) characterization of the quantum set of correlations~\cite{TPBA24}, either with~\cite{Brunner08,DNPA14,DNPA15,NFA15} or without~\cite{NPA,NPA2008,Doherty08,Moroder13} a {\em dimension constraint}, the above observation leads to a significant reduction~\cite{Moroder13,rosset2018symdpoly,TRR19,Ioannou:2021qmh} in the number of optimization parameters. In finding an explicit QS that realizes this maximal violation, note, however, that this reduction to symmetry QSs may come at a price of an increase~\cite{Moroder13} in the required HSD.
    
Will it be possible to enjoy this symmetry reduction while keeping the required QS at its minimal HSD? Notice that correlations respecting any given symmetry represent a strict subset of all possible correlations. In this context, {\em asymmetry} is evidently a resource for Bell-inequality violation. Hence, if the answer to this question is negative for any given Bell inequality, there exists a {\em trade-off} between two kinds of resources---HSD and asymmetry---that one may employ to maximize its quantum violation. Here, we systematically explore and answer this question for Bell inequalities defined in several (mostly) bipartite Bell scenarios, including those with binary outcomes and up to four alternative measurements, as well as those with an arbitrary number of outcomes but only binary measurement choices.
   
We structure the rest of this paper as follows.	 In~\cref{Sec:Prelim}, we introduce our notations, recapitulate essential notions of Bell nonlocality~\cite{Bell-RMP}, and provide a more formal explanation of the notion of a minimal QS. After that, in~\cref{Sec:Symmetry}, we introduce the various definitions related to the symmetry of PPI and remind specifically in \cref{Sec:SymPurStrategy} how a QS giving rise to a symmetric correlation can {\em always} be converted into a purified, symmetric QS. Examples of Bell inequalities that can be maximally violated using a symmetric QS in the minimal dimension are then provided in \cref{sec:Minsymmetricstraetegy}. In contrast, examples where a trade-off exists are presented in~\cref{sec:Asymmetriccase}. In~\cref{Sec:Asym2Sym}, we give a general discussion of asymmetric QSs giving rise to symmetric correlations. Then, we discuss in \cref{Sec:GeometryandSymmetry} the implications of some of these examples on the geometry of the set of quantum correlations and self-testing.  Finally, we conclude in~\cref{Sec:Conclusion}. We provide further details about the numerical methods employed in \cref{App:Techniques} and other miscellaneous results in \cref{App:MISC}.

\section{Preliminaries}\label{Sec:Prelim}
    \subsection{Correlations in Bell scenarios}

    Consider the bipartite Bell scenario $(2,m,n)$, in which two parties, Alice and Bob, can both perform $m$ measurements, each resulting in $n$ outcomes.  Let $\intv{n}:=\{0,1,\cdots,n-1\}$. We label Alice's and Bob's settings by $x, y\in[m]$ and their outcomes by $a, b\in[n]$, respectively. The statistics of a Bell test yield the joint probability distribution (or \textit{correlation}) $\vecP_\AB = \{\PAB(a,b|x,y)\}$,  which manifests how well their measurement outcomes correlate. Throughout, when there is no risk of confusion, the subscripts $_\AB$ will be omitted for simplicity.

    We say that a correlation is \textit{Bell-local} (hereafter abbreviated as \textit{local}) if it admits an LHV description~\cite{Bell64}:
    \begin{equation}
    \label{eq:belllocality}
        \PAB(a,b|x,y)\overset{\mathcal{L}} = \sum_{\lambda}q(\lambda)f_\A(a|x,\lambda)f_\B(b|y,\lambda),
    \end{equation}
    where $\lambda$ is the local-hidden variable (equivalently, shared randomness), $q(\lambda)$ is its distribution weight, and $f_i(\cdot|\cdot,\lambda)=0,1$ with $i\in\{A,B\}$ are local response functions. Note that $\L$---the set of local correlations---forms a convex polytope, called the Bell polytope, which consists of finitely many extreme points, each corresponding to a local deterministic strategy. 
    
    When the outcomes are binary, i.e., $n=2$, we can also conveniently express a correlation $\vecP$ through the expectation values (or \textit{correlators}) $\langle A_x \rangle$, $\langle B_y \rangle$, and $\langle A_xB_y \rangle$, defined as:
    \begin{equation}\label{Eq:Correlators}
        \begin{aligned}
            &\langle A_x \rangle = \sum_{a=0,1}  (-1)^{a} P_\A(a|x), \,\, \langle B_y \rangle = \sum_{b=0,1}  (-1)^{b} P_\B(b|y), \\
            &  
            \langle A_xB_y \rangle =  \sum_{a,b=0,1}(-1)^{a+b}\PAB(a,b|x,y),
        \end{aligned}
    \end{equation}
    where $A_x$ ($B_y$) is the outcome of Alice (Bob) for her (his) $x$-th ($y$-th) measurement, while $P_\A$ and $P_\B$ are obtained from $\PAB$ via the marginalization over $B$ and $A$, respectively. For example, $P_\A(a|x) = \sum_b P_\AB(a,b|x,y)$ for all $a,x$. 
    
    In contrast with those compatible with an LHV description, a correlation $\vecP$ is said to be \textit{quantum} if it arises from locally measuring a shared quantum state $\rho_\AB$, say, acting on $\mathbb{C}^D\otimes\mathbb{C}^D$. Explicitly, a quantum correlation $\vecP$ associated with a {\em quantum strategy} (QS) 
    \begin{equation}\label{Eq:QS}
    	\Qstr=\{\rho_\AB,\{M_{a| x}^{A}\}_{a,x},\{M_{b|y}^{B}\}_{b,y}\}
    \end{equation}	 
	of local HSD upper bounded by $D$ may be computed from Born's rule as:
    \begin{equation}
    \label{Eq:Born}
        \PAB(a,b\vert x,y) 
        \overset{\mathcal{Q}}{=} \mathrm{tr}(\rho_\AB M_{a\vert x}^{A}\otimes M_{b\vert y}^{B}).
    \end{equation}
    Here, $\{M_{a\vert x}^{A}\}_a$ (respectively $\{M_{b\vert y}^{B}\}_b$) is a positive operator-valued measure (POVM) for Alice's $x$-th (Bob's $y$-th) measurement, i.e., $\sum_a M_{a\vert x}^{A} = \id_D$ and $M_{a|x}^A\succeq 0$ for all $a,x$ where $\id_D$ is the identity operator acting on $\mathbb{C}^D$. Henceforth, we denote by 
    $\Q$ the set of quantum correlations (for any given Bell scenario). A celebrated discovery by Bell~\cite{Bell64} is that not all $\vecP\in\Q$ can be cast in the form of \cref{eq:belllocality}, viz. some quantum $\vecP$ lies outside of $\L$.

    Mathematically, a Bell polytope $\L$ can equivalently be described in terms of a minimal set of halfspaces, called Bell inequalities.
    In the $(2,m,n)$ Bell scenario, a general linear Bell inequality reads as:
    \begin{equation}
    \label{eqn:bellineq}
        I_{\vec{\beta}}
        = \vec{\beta} \cdot \vec{P}
        = \sum_{x,y=0}^{m-1}\sum_{a,b=0}^{n-1}\beta_{ab}^{xy}\PAB(a,b\vert x,y)\overset{\mathcal{L}}{\le} L_{\vec{\beta}},
    \end{equation}
    where $\vec{\beta} = \{\beta_{ab}^{xy}\}$ is a vector of real coefficients that determine the corresponding local bound $L_{\vec{\beta}}$. By definition, a Bell inequality, cf. \cref{eqn:bellineq}, is satisfied by {\em all} $\vecP\in\L$. Hence, the fact that a $\vecP\not\in\L$, i.e., the correlation $\vecP$ is \textit{nonlocal}, can be witnessed from its {\em violation} of a Bell inequality.
    
    As an explicit example, recall from~\cite{Clauser69} the simplest nontrivial Bell inequality, i.e., the CHSH Bell inequality, which admits a PPI form: 
        \begin{equation}\label{Eq:CHSH}
                    \ICHSH = \sum_{x,y=0}^{1}\sum_{a,b=0}^{1} (-1)^{xy+a+b} \PAB(a,b\vert x,y)
                    \overset{\mathcal{L}}{\le} 2.
        \end{equation}
    Often, this is also written, in the notations of~\cref{Eq:Correlators}, as
       \begin{equation}\label{Eq:CHSH2}
           \begin{aligned}
               \ICHSH =\langle A_0 B_0\rangle+\langle A_0 B_1\rangle+\langle A_1 B_0\rangle-\langle A_1 B_1\rangle \overset{\L}{\le} 2.
           \end{aligned}
       \end{equation}

    \subsection{Minimal QS for maximal quantum violation}

	For any given Bell inequality $I$, it is natural to wonder the extent to which it can be violated quantum-mechanically and identify the QSs that result in this maximal violation. In particular, from a resource-theoretic perspective, it is of interest to determine the {\em minimal}, i.e., the {\em smallest} local HSD $\Dm$ capable of realizing such a quantum violation. For concreteness, we refer to any QS in $\Dm$ that realizes the maximal quantum violation of a Bell inequality as a {\em minimal maximizing QS}, which we denote by $\mMQS$.
	
	For example, for Bell inequalities defined in an $(N,2,2)$ Bell scenario (with $N$ being any integer larger than or equal to $2$), it is known~\cite{Masanes06} that $\Dm=2$ and we may take $\mMQS$ to be an $N$-qubit pure state along with projection-valued measures (PVMs), i.e., projective measurements. In particular, a well-known two-qubit strategy achieving the maximal violation of the CHSH Bell inequality \cite{Clauser69} of \cref{Eq:CHSH2} consists of the following state and observables:
    \begin{equation}\label{eq:max_CHSH_strategy}
        \begin{gathered}
            \ket{\Phi^+}_\AB = \frac{1}{\sqrt{2}}[{\ket{00}}+\ket{11}],\\
            \hat{A}_0 = \sigma_z,\quad \hat{A}_1 = \sigma_x,\\
            \hat{B}_k = \frac{1}{\sqrt{2}} \big[ \sigma_z + (-1)^k \sigma_x \big],\quad k=0,1,
        \end{gathered}
    \end{equation}
    where $\sigma_i$ with $i\in\{x,y,z\}$ are Pauli matrices, while $\hat{A}_x$ and $\hat{B}_y$ are  Alice's  $x$-th and Bob's $y$-th observable, which are related to their POVMs by $\hat{A}_x=\sum_{a} (-1)^aM^A_{a|x}$ and $\hat{B}_y=\sum_{b} (-1)^bM^B_{b|y}$.

    Beyond CHSH, surprisingly little is known about the $\Dm$ of various Bell inequalities. Results from \cite{ADG+02}, \cite{NPA2008}, and \cite{Ioannou:2021qmh} show that for the family of CGLMP inequalities~\cite{CGLMP} $I_d$ with integer $d\le 8$, their minimal dimension $\Dm\le d$. For the specific case of $d=3$, a dimension bound deduced from a negativity~\cite{VW02} lower bound~\cite{Moroder13}  further shows that this bound is tight, likewise for the cases of $d=4$ and $5$, as we show in \cref{App:CGLMP}.
    
    On the other hand, for the $I_{3322}$ inequality from~\cite{Collins04}, its $\Dm$ could well be {\em infinite}, see~\cite{PV10}. Beyond these, the maximal Bell violation of various other Bell inequalities in the bipartite~\cite{Pal09QB,Liang09,Schwarz16,Zambrini19_4422} and multipartite~\cite{Grandjean2012} Bell scenarios has also been investigated. Again, each explicit QS attaining the quantum maximum provides an upper bound on the corresponding $\Dm$. Some other, more general upper bounds on $\Dm$ have also been recently established in \cite{PAC+2025}.

    \section{The symmetry of party-permutation invariance}
    \label{Sec:Symmetry}

     In Moroder \textit{et al.}~\cite{Moroder13}, it has been shown that {\em any} PPI quantum correlation, which is sufficient for maximizing the quantum value of a PPI Bell inequality, can always be realized by a QS involving a PPI state and the same set of local measurements performed by each party. Consequently, when there is no restriction in HSD, the maximal quantum violation of a PPI Bell inequality is always attainable~\cite{Moroder13} using a PPI QS. In this section, we introduce several definitions pertaining to the symmetry of PPI, which is the only symmetry considered in this work.
     
    \subsection{Symmetric correlations, Bell inequalities, and strategies}

    \begin{definition}[Symmetric correlation]
    \label{def: Symmetric correlation}
        A correlation $\vecP_\AB$ arising from a bipartite scenario is said to be \textit{symmetric} if it remains invariant under the exchange of parties, $\A\leftrightarrow\B$, i.e., the simultaneous exchange of the parties' settings $x\leftrightarrow y$ and outcomes $a\leftrightarrow b$:
    \begin{equation}\label{Eq:Psym}
            \PAB(a,b\vert x,y) = \PAB(b,a\vert y,x),  \quad \forall\,\, a, b, x, y.
    \end{equation}
    \end{definition}
    
    \begin{definition}[Symmetric Bell inequality]
    \label{def:Symmetric Bell inequality}
        Consider a Bell inequality, cf.~\cref{eqn:bellineq}, characterized by the coefficients, $\vec{\beta} =\{\beta^{ab}_{xy}\}$.
        We say that a Bell inequality is {\em symmetric} if the coefficients satisfy 
            \begin{equation}\label{eq:sym_ineq}
                \begin{aligned}
                    &\beta^{ab}_{xy} = \beta^{ba}_{yx},
                    && \forall\,\, a,b,x,y.
                \end{aligned}
            \end{equation}
    \end{definition}
    Throughout, we use $\Psym$ ($\BIsym$) to denote a symmetric correlation (Bell inequality) whenever we want to emphasize that it satisfies \cref{Eq:Psym} [\cref{eq:sym_ineq}].  Notably, synchronous correlations~\cite{Paulsen16_ChromaticNumber} arising from the maximally entangled states~\cite{RL17} are symmetric. Synchronous correlations are characterized by the extra constraint: $\PAB(a,b\neq a|x,x) = 0, \forall\,\, x$, which means that in the context of a nonlocal game~\cite{CLeve2004}, the two players must return the same answer upon receiving identical inputs.
   
    As originally remarked in~\cite{Moroder13}, symmetric correlations are particularly relevant in maximizing the (quantum) violation of a symmetric Bell inequality.
    \begin{proposition}\label{Prop:SufficiencySymP}
        In maximizing the (quantum) value of a symmetric Bell inequality, it suffices to consider a symmetric (quantum) correlation.
    \end{proposition}
    \begin{proof}
    For simplicity, we provide a proof applicable to the bipartite quantum case. Let $I$ be a symmetric Bell inequality specified by $\BIsym$ and $\vecP_\AB^\star$ be a quantum maximizer of $I$ that is {\em asymmetric}. The asymmetric nature of $\vecP_\AB^\star$ implies that upon the action of the permutation $V_\AB:\A\leftrightarrow\B$, we obtain $\vecP_\BA:=V_\AB\vecP_\AB^\star\neq\vecP_\AB^\star$. Since $\vecP_\AB^\star$ is a maximizer of $I$ and $\BIsym$ is symmetric (i.e., $V_\AB\BIsym = \BIsym$), we see that 
		\begin{equation}
    		\begin{aligned}
    			\max_{\vecP\in\Q} I &= \max_{\vecP\in\Q} \BIsym\cdot\vecP = \BIsym\cdot\vecP_\AB^\star \\
    			&= (V_\AB\BIsym)\cdot(V_\AB\vecP_\AB^\star) = \BIsym\cdot\vecP_\BA,
    		\end{aligned}
		\end{equation}
		where the third equality follows from the fact that a simultaneous permutation on the Bell coefficients $\BIsym$ and the correlation $\vecP_\AB^\star$ leaves their inner product unchanged.
		Hence, $\vecP_\BA$ must also be a maximizer of $I$. Finally, using the fact that $I$ is linear in $\vecP$, we see that the symmetrized quantum correlation
		\begin{equation}\label{Eq:SymmetrizingP}
			\Psym:= \frac{1}{2}(\vecP_\AB+ \vecP_\BA)
		\end{equation}
		must also violate $I$ maximally.
    \end{proof}
	      
    Let us now recall from \citep[Proposition 1]{Moroder13} a particular kind of QSs that yield a symmetric correlation $\Psym$.
    \begin{definition}[Symmetric QS]
    \label{def:Symmetric strategy}
        In a bipartite Bell scenario, a symmetric QS  (abbreviated as SQS), 
        is a QS where both parties perform the same measurements (i.e., $M^{B}_{a|x} = M^{A}_{a|x}$  $ \forall$ $a,x$) on a shared PPI state $\rho_\BA  = \rho_\AB$, where 
        \begin{equation}\label{Eq:rhoBA}
        		\rho_\BA := S\rho_\AB S^{\dagger}, 
	\end{equation}		
    and 
    \begin{equation}\label{Eq:Swap}
	S:= \sum_{ij} \ketbra{i}{j}_\text{A}\otimes \ketbra{j}{i}_\text{B}
    \end{equation}
    is the swap (unitary) operator and $\{\ket{i}\}_i$ is an orthonormal set of basis vectors.
    \end{definition}
    
    From \cref{Eq:Born} and \cref{def:Symmetric strategy}, one can see that an SQS {\em must} give rise to a symmetric quantum correlation $\Psym$:
    \begin{equation}\label{eqn: symmetric strategy}
        \begin{aligned}
            \PAB(a,b\vert x,y) 
            &= \mathrm{tr}(\rho_\AB M_{a\vert x}^{A}\otimes M_{b\vert y}^{B})\\
            &= \mathrm{tr}(S\rho_\AB S^\dag SM_{a\vert x}^{A}\otimes M_{b\vert y}^{B}S^\dag)\\
            &= \mathrm{tr}(\rho_\BA M_{b\vert y}^{B}\otimes M_{a\vert x}^{A})\\
            &= \mathrm{tr}(\rho_\AB M_{b\vert y}^{A}\otimes M_{a\vert x}^{B})\\
            &= \PAB(b,a\vert y,x),
        \end{aligned}
    \end{equation}
    where we have used the unitarity of $S$ and the cyclic property of trace to arrive at the second equality, \cref{Eq:rhoBA,Eq:Swap} to arrive at the third equality, and \cref{def:Symmetric strategy} to arrive at the fourth equality. Thus, an SQS yields a symmetric correlation, cf. \cref{def: Symmetric correlation}.   
      
    Nevertheless, a $\Psym$ is not necessarily realized by an SQS. For example, the maximal-CHSH-violating two-qubit strategy of~\cref{eq:max_CHSH_strategy}---with the two parties performing different measurements---is clearly not symmetric, even though it gives the evidently symmetric Tsirelson correlation $\vec{P}_T$ \cite{Tsirelson1980}:
    \begin{equation}\label{eq:Tsirelson pt}
            \vec{P}_T \coloneqq 
            \left\{ 
                \PAB(a,b|x,y) = \frac{1}{4}
                \left[ 
                    1 + \frac{(-1)^{xy+a+b}}{\sqrt{2}}
                \right]
            \right\}.
    \end{equation}

    \subsection{Symmetrization and purification of a QS} 
    \label{Sec:SymPurStrategy}     
      
    While a symmetric quantum correlation $\Psym$ need not originate from an SQS, when there is no restriction in the Hilbert space dimension,  any {\em asymmetric} QS giving a symmetric correlation can always be transformed into an SQS producing the same $\Psym$. We now recall from \cite{Moroder13} such a transformation. Given an arbitrary $\Qstr=\{\rho_\AB ,\{M_{a|x}^{A}\}_{a,x},\{M_{b|y}^{B}\}_{b,y}\}$ realizing a symmetric correlation, its \textit{symmetrized} version \begin{equation}\label{Eq:SQS}
      		\SQS:=\{\lrvec{\rho}_\AB , \{N_{a|x}\}_{a,x}, \{N_{b|y}\}_{b,y}\}
    \end{equation}	 
	can be constructed by introducing ancillary projectors $\{\ketbra{0}{0},\ketbra{1}{1}\}$ associated with each party (see~\citep[Eqs.~(17,18)]{Moroder13}):
    \begin{align}\label{eq:mixed-symmetrization}
              \lrvec{\rho}_\AB  &= \frac{1}{2} \l [ \rho_\AB  \ten \ket{01}_{\A'\B'}\!\bra{01} +\rho_\BA   \ten \ket{10}_{\A'\B'}\!\bra{10} \r], \nonumber\\
              N_{a|x} & = M_{a\vert x}^{A} \ten \ketbra{0}{0} + M_{a\vert x}^{B} \ten \ketbra{1}{1},
    \end{align}
    where $\A'$ ($\B'$) is a label of Alice's (Bob's) ancillary space.
    By construction, both parties employ the same POVMs given by $\{N_{a|x}\}_{a|x}$, whilst the PPI nature of $\lrvec{\rho}_\AB$ can also be straightforwardly verified.
	To see that both $\Qstr$ and $\SQS$ give the same symmetric correlation, let us denote by $\lrvec{P}_\AB(a,b|x,y)$ components of the correlation arising from \cref{Eq:SQS} and notice from~\cref{Eq:Born,eq:mixed-symmetrization} that
    \begin{equation}
            \begin{aligned}
            	&\lrvec{P}_\AB(a,b|x,y) =\ \tr \Big[ \big( N_{a|x} \ten N_{b|y} \big) \lrvec{\rho}_\AB \Big]\\
                =\ &\frac{1}{2} \tr \Big[ \big( M^A_{a|x} \ten M^B_{b|y} \big) \rho_\AB + \big( M^B_{a|x} \ten M^A_{b|y} \big) \rho_\BA\Big]\\
                =\ &\frac{ \PAB(a,b|x,y) + \PAB(b,a|y,x) }{2} 
                = \PAB(a,b|x,y),
            \end{aligned}
    \end{equation}
    where the last equality follows from the {\em assumed} symmetry of $\vecP$.
    Thus, this symmetrization embeds the original $\Qstr$ acting on $\HS\otimes\HS$ to an SQS acting on $[\HS\ten \mathbb{C}^2]\otimes[\HS\ten \mathbb{C}^2]$, doubling the original local HSD, while preserving the produced correlation $\vecP$. 
    Importantly, as we see below, the symmetrization of a QS may also be achieved, in some cases, via a local unitary transformation.
    
    On the other hand, it is also well-known that for any given $\Qstr$ realizing a correlation $\vecP$, one can obtain, through Naimark dilation (\citep[Section 9]{HP16Notes}) and quantum state purification, a {\em purified} QS (PQS)
    \begin{equation}\label{Eq:PQS}
	\begin{gathered}
    	\Qstr'=\{\ket{\psi}_\AB\bra{\psi}, \{\Pi_{a|x}^{A}\}_{a,x},\{\Pi_{b|y}^{B}\}_{b,y}\},\\
		(\Pi_{a|x}^{A})^2=\Pi_{a|x}^{A}\,\,\forall\,\,a,x,\quad (\Pi_{b|y}^{B})^2=\Pi_{b|y}^{B}\,\,\forall\,\,b,y,
	\end{gathered}	
    \end{equation} 
    that realizes $\vecP$. 
    In particular, if the state defined in $\Qstr$ is a (multipartite) PPI density operator, one can follow the method described in~\citep[Section 4.2]{Renner06_SecurityQKD} to purify it into a pure state lying in the {\em symmetric subspace}. Hence, by concatenating the symmetrization procedure of~\cite{Moroder13} and the purification procedure, one can always obtain a purified SQS (PSQS) realizing any given $\Psym\in\Q$. Alternatively, from $\Qstr'$, we can also obtain a PSQS 
    \begin{equation}\label{Eq:PSQS}
    	\tilde{\Qstr}:=\{\tilde{\ket{\phi}}_\AB\tilde{\bra{\phi}}, \{\tilde{\Pi}_{a|x}\}_{a,x}, \{\tilde{\Pi}_{b|y}\}_{b,y}\}
	\end{equation} 
	via
    \begin{equation}\label{eq:pure-symmetrization}
        \begin{aligned}
            \tilde{\ket{\phi}}_\AB 
            & = \frac{1}{\sqrt{2}} \l [ \ket{\psi}_\AB  \ket{01}_{\A'\B'} +\ket{\psi}_\BA    \ket{10}_{\A'\B'} \r], \\
              \tilde{\Pi}_{a|x}
              & = \Pi_{a\vert x}^{A} \ten \ketbra{0}{0} + \Pi_{a\vert x}^{B} \ten \ketbra{1}{1}.
        \end{aligned}
    \end{equation}
    
    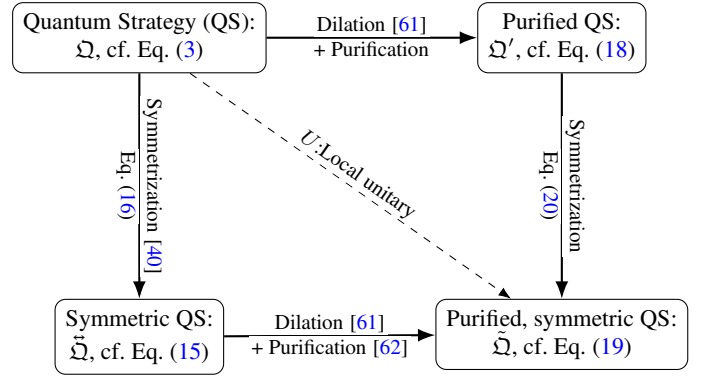
\begin{figure}[t!]
        \centering
            \begin{tikzpicture}[
      node distance=3.0cm and 2.8cm,
      every node/.style={font=\small},
      box/.style={draw, rounded corners, align=center, inner sep=4pt, minimum width=2.0cm},
      lab/.style={font=\footnotesize, inner sep=1pt, fill=white, fill opacity=.0, text opacity=1},
      arr/.style={-Latex, thick},
      darr/.style={-Latex, dashed},
      dotarr/.style={-Latex, dotted}
    ]

    \node[box] (Q) {Quantum Strategy (QS):\\  $\Qstr$, cf.~\cref{Eq:QS}};
    \node[box, right=of Q] (UV) {Purified QS:\\ $\Qstr'$, cf.~\cref{Eq:PQS}};
    \node[box, below=of Q] (SQS) {Symmetric QS:\\ $\SQS$, cf.~\cref{Eq:SQS}};
    \node[box, below=of UV] (PSQS) {Purified, symmetric QS:\\ $\tilde{\Qstr}$, cf.~\cref{Eq:PSQS}};
    
    \draw[arr] (Q) -- (UV)
      node[midway,above,lab] {Dilation~\cite{HP16Notes}}
      node[midway,below,lab] {+ Purification};
    
    \draw[arr] (Q) -- (SQS)
      node[midway,lab,sloped,above] {Symmetrization~\cite{Moroder13}}
      node[midway,lab,sloped,below] {\cref{eq:mixed-symmetrization}};

    \draw[arr] (SQS) -- (PSQS)
      node[midway,above,lab] {Dilation~\cite{HP16Notes}}
      node[midway,below,lab] {+ Purification~\cite{Renner06_SecurityQKD}};
    
    \draw[arr] (UV) -- (PSQS)
      node[midway,lab,sloped,above] {Symmetrization}
      node[midway,lab,sloped,below] {\cref{eq:pure-symmetrization}};
      
    \draw[darr] (Q) -- (PSQS)
      node[pos=0.5,above,lab, sloped] {$U$:Local unitary};
    \end{tikzpicture}
    \caption{\label{Fig:TowardPSQS} Schematic showing different pathways to obtain a purified, symmetric quantum strategy (PSQS) $\tilde{\Qstr}$ from any quantum strategy (QS) $\Qstr$ producing a symmetric correlation $\Psym$. In general, this involves performing step 1. Naimark dilation + purification (horizontal solid arrow), and 2. symmetrization (vertical solid arrow) in either order. In both cases, the local HSD is at least {\em doubled}. However, if the initial strategy $\Qstr$ consists of a pure state and PVMs, it may even be possible to obtain a PSQS via a local unitary transformation (dashed arrow), which preserves the local HSD: an example being the transformation of the strategy of \cref{eq:max_CHSH_strategy} to that of \cref{eq:CHSH ss strategy}.}
    \end{figure}
        
    Hence, we have shown that the following Proposition holds.
    \begin{proposition}\label{Prop:PurifiedSymStra2}
        In a bipartite Bell scenario, a symmetric correlation, see~\cref{def: Symmetric correlation}, can always be realized using an SQS consisting of a PPI bipartite pure state and with both parties performing the same local PVMs.
    \end{proposition}

    Notice that Proposition~\ref{Prop:PurifiedSymStra2} can be seen as a strengthening of the observation given in~\cite{Moroder13}, in that we may not only take the QS reproducing {\em any} given $\Psym$ to be symmetric, but also purified. Evidently, \cref{def: Symmetric correlation,def:Symmetric Bell inequality,def:Symmetric strategy} can be naturally generalized to an $N$-partite Bell scenario (with $N>2$) by demanding PPI for all possible permutations of parties. In this regard, we remark that the above symmetrization procedures, and hence Proposition~\ref{Prop:PurifiedSymStra2}, can be generalized to give the following result.
    \begin{proposition}\label{Prop:SCor2SQS}
    	In a multipartite Bell scenario, a PPI correlation can always be realized using a PPI QS consisting of a pure state lying on the {\em symmetric subspace} and with all parties performing the same local PVMs.
    \end{proposition}

    In \cref{Fig:TowardPSQS}, we provide a summary of the different pathways leading from a general QS for a PPI correlation $\Psym$ to a PSQS. Note that although~\cite{Moroder13} asserts that the multipartite QS realizing a PPI correlation may be symmetrized, an explicit procedure was only given for the bipartite scenario. In contrast, an explicit procedure for generating a strategy that respects translational invariance (rather than PPI) in the multipartite scenario was provided in~\cite{Tura_2014}. For completeness, we provide in~\cref{App:Multipartite}  the explicit symmetrization procedure corresponding to Proposition \ref{Prop:SCor2SQS}, which generalizes~\cref{Eq:PSQS,eq:pure-symmetrization} beyond the bipartite scenario.

\section{Examples where SQS in the minimal dimension can be maximizing}
    \label{sec:Minsymmetricstraetegy}
    We now give some examples of {\em symmetric} Bell inequalities where we observe no trade-off between symmetry and dimension in achieving their maximal violation. In other words, for these inequalities, one can indeed find a $\mMQS$ that is also an SQS, which we shall denote by $\mMSQS$.

    \subsection{The CHSH Bell inequality}
        \label{sec:symmetricCHSH}
    
    Our first example involves the CHSH Bell inequality of \cref{Eq:CHSH}, which is---modulo the freedom in relabeling~\cite{Collins04} of settings and outcomes--- known~\cite{Fine:PRL:1982} to be the only nontrivial facet-defining Bell inequality in the simplest bipartite Bell scenario.
    The maximal quantum violation of CHSH can be attained using the asymmetric QS defined in~\cref{eq:max_CHSH_strategy}, giving the {\em symmetric} Tsirelson correlation $\vecP_T$ of \cref{eq:Tsirelson pt} and the so-called Tsirelson bound~\cite{Tsirelson1980} of $2\sqrt{2}$. However, as mentioned  
    in~\citep[Section IV.A]{Bancal15} (see also~\cite{Xingyao16}), this quantum value of $\ICHSH$ can also be obtained using an SQS.
    
    To this end, it suffices
    \footnote{
    Note that the same transformation is obtained by first rotating by 
    $-\frac{\pi}{4}$ about the $y$-axis, then by $\pi$ about the $z$-axis on the Bloch sphere.
    } 
    to apply on Bob's qubits a $\theta = \pi$ rotation $R_{\hat n}(\theta) = \exp(-\frac{i \theta}{2} \hat n \cdot \vec\sigma)$ about the $\hat n$-axis of the Bloch sphere, where $\hat n = (\sin\frac{\pi}{8},0,\cos\frac{\pi}{8})$ and  $\vec \sigma = (\sigma_x, \sigma_y, \sigma_z)$.
    The resulting shared state, which is indeed symmetric, can be written as
    \begin{subequations}\label{eq:CHSH ss strategy}
        \begin{equation}\label{Eq:CHSHSymState}
            \begin{aligned}
                \ket{\psi}_\AB  
                &= \mathbb{I}_A \ten R_{\hat n}(\pi)\ket{\Phi^+}_\AB \\
                &= -i\Big[ \cos(\frac{\pi}{8})\ket{\Phi^-} +\sin(\frac{\pi}{8})\ket{\Psi^+}\Big],
            \end{aligned}
        \end{equation}
        where $\ket{\Phi^-} = \frac{1}{\sqrt{2}}[\ket{00}-\ket{11}]$ and $\ket{\Psi^+} = \frac{1}{\sqrt{2}}[\ket{01}+\ket{10}]$.
        Moreover, after the same unitary transformation, Alice's and Bob's observables become
        \begin{equation}\label{Eq:CHSHSymObs}
            \hat{A}_0=\hat{B}_0=\sigma_z, \quad \hat{A}_1=\hat{B}_1=\sigma_x,
        \end{equation}
    \end{subequations}
    which are clearly symmetric. 
    As we see next in~\cref{CGLMP}, the SQS of \cref{eq:CHSH ss strategy} can be seen as the $d=2$ instance of a more general family of SQSs.
        
    \subsection{Inequalities in the $(2,2,n)$ Bell scenarios}\label{CGLMP}

    The CGLMP inequality~\cite{CGLMP} (see also~\cite{Kaszlikowski02}), denoted by $I_{d}$, constitutes a facet-defining inequality of the Bell polytope~\cite{Masanes02_Tight_CGLMP} in the $(2,2,n)$ Bell scenarios\footnote{For $n\ge 3$, the outcome-lifted~\cite{Pironio_Lifting_05} CHSH is another facet-defining Bell inequality in this scenario, but we know from the results of~\cite{Jeba:PRR:2019} that it cannot exhibit a trade-off.} for $d=n\ge 2$. Specifically, for $d=2$, it may be rewritten as the CHSH Bell inequality of~\cref{Eq:CHSH2} after the relabeling $B_1\leftrightarrow B_2$. Although $I_d$ in its original form~\cite{CGLMP} is not PPI, it can be recast into a symmetric form (cf. \cref{def:Symmetric Bell inequality}) by, e.g.,  relabeling all of Alice's outcomes from $a$ to $d-a$. 
    
    Alternatively, as demonstrated in~\citep[Appendix B.1.1]{Liang:PhDthesis}, by applying the no-signaling conditions~\cite{Popescu:FP:1994,Barrett_05}, normalization constraints, and appropriate relabeling of the outcomes, the CGLMP inequality can be converted to the $I_{22dd}$ inequality~\cite{Collins04}, which is manifestly symmetric:
    \begin{align}\label{eq: symmetry 22dd}
        I_{22dd} 
        &= \sum_{a,b\in\triangle_\le}  \PAB(a,b|0,0)- \sum_{a,b\in\triangle_\ge} \PAB(a,b|1,1) \nonumber\\
        &+ \sum_{a,b\in\triangle_\ge} \left[ \PAB(a,b|0,1) + \PAB(a,b|1,0)  \right] \\
        &- \sum_{a=0}^{d-2} \PA(a|0)
        - \sum_{b=0}^{d-2} \PB(b|0) \overset{\L}{\le} 0,\nonumber
    \end{align}
    where $\triangle_\le$ and $\triangle_\ge$ refer, respectively, to the set of $a,b\in\intv{d-1}$ such that $a+b\le d-2$ and $a+b\ge d-2$.
    Moreover, for any given $\Qstr$, the Bell value of $I_d$ and that of $I_{22dd}$ are related~\cite{Liang:PhDthesis} by 
    \begin{equation}
        \label{eq:22ddequ}
        I_{d} = \frac{2d}{d-1}I_{22dd} + 2.
    \end{equation}
    
    The best-known Bell violation of the CGLMP inequality $I_d$ (or equivalently $I_{22dd}$) can be achieved by locally measuring a partially entangled two-qudit state~\cite{ADG+02} in the judiciously chosen ``Fourier-transformed'' bases~\cite{CGLMP} (implementable via a multiport beam splitter in a photonic setup, see, e.g.,~\cite{KGZ+00,DKZ01}).  Numerical optimizations from~\cite{ADG+02,NPA2008} confirm this QS to be optimal for dimensions $2\le d \le 8$. In particular, for $d < 5$, a sum-of-squares decomposition in the symmetric form provides an analytic bound on the maximal quantum violation~\cite{Ioannou:2021qmh}.
    
    However, the optimal measurements considered in~\cite{CGLMP} are not PPI, even after we incorporate the outcome-relabeling needed to cast $I_d$ in a PPI form.
    In the following, we present a family of SQSs that match the best known $I_{22dd}$ (and hence $I_d$) Bell-inequality violation. 
    For this purpose, we specify the POVMs (with $M_{a|x}^{B} = M_{a|x}^{A}\,\,\forall\,\,a,x$) as:
    \begin{subequations}\label{Eq:CGLMP-Opt}
        \begin{equation} 
            M_{a|0}^{A} = \ketbra{a}{a}, \quad M_{a|1}^{A} =  U\ketbra{a}{a}U^{\dagger},\quad U:=TW,
        \end{equation}
        where $\{\ket{a}\}_{a =0}^{d-1}$ is the computational basis, and the $d\times d$ unitary operator $U$ is defined via\footnote{To see that the operator $U$ define as such is indeed unitary, see~\cref{App:Unitarity}.}
        \begin{equation}\label{Eq:TW}
            T:= -1 \bigoplus \mathbb{I}_{d-1} ~\text{  and  }~
            W:= \!\!\!\sum_{i,j\in\intv{d}}\frac{\ket{i}\bra{j}}{d\sin{[(i-j-\frac{1}{2})\frac{\pi}{d}}]}. 
        \end{equation}
    \end{subequations}        
    Accordingly, the quantum state to be measured $\ket{\psi_d}$ is obtained by solving the eigenvector corresponding to the largest eigenvalue of the symmetric $I_{22dd}$ Bell operator~\cite{Braunstein1992}.
    
    As an explicit example, the optimal QS for $d=3$ consists of measuring
    \begin{align}
        \ket{\psi_3}=
        \gamma
        &\Big[ \alpha(\ket{00}-\ket{12}-\ket{21}) +\left(\ket{01}+\ket{10}+\ket{22}\right)\nonumber\\
        &\left. -\frac{\alpha+1}{2}(\ket{02}+\ket{20}-\ket{11})\right],
    \end{align}
    where $\alpha = \frac{5-\sqrt{33}}{2}$ and $\gamma=\frac{2}{3}\sqrt{\frac{2}{55-9\sqrt{33}}}$ is a normalization factor. Modulo a local change of basis, $\ket{\psi_3}$ is exactly the optimal state for $I_{d=3}$ found in \cite{ADG+02} (see also~\cite{Bancal15}).
    More generally, the optimal quantum state obtained via the Bell operator defined by the PPI measurement strategy of \cref{Eq:CGLMP-Opt} is easily seen to be PPI, too. Furthermore, for $d=2$ to $19$, our computation shows that---to six or more significant digits---the SQS formed by these states and measurements indeed gives an $I_{22dd}$ (and hence $I_{d}$) value that coincides with the upper bound obtained from the SDP hierarchy of Navascu\'es-Pironio-Ac\'in (NPA). See~\cref{App:CGLMP} for details.
    
    Naturally, given the above observation, one may wonder whether the same conclusion also holds for a special family of the Salavrakos-Augusiak-Tura-Wittek-Ac\'{\i}n-Pironio~\cite{SAT+17} (SATWAP) Bell inequalities that are applicable in the $(2,m,d)$ Bell scenarios, known to be violated maximally using the optimal measurements for the CGLMP inequality. Here, we focus on the $m=2$ case of this family of inequalities, which was extensively studied in~\cite{Sarkar:2021aa}. 
    
    To this end, following~\cite{SAT+17,Sarkar:2021aa}, it is expedient to generalize the definition of correlators from~\cref{Eq:Correlators} to
    \begin{equation}\label{Eq:Correlators:d}
        \begin{aligned}
            &\langle A^k_x \rangle = \sum_{a=0}^{d-1}  \omega^{ak} P_\A(a|x), \,\, \langle B^\ell_y \rangle = \sum_{b=0}^{d-1}  \omega^{b\ell} P_\B(b|y), \\
            &  
            \langle A^k_xB^\ell_y \rangle =  \sum_{a,b=0}^{d-1} \omega^{ak+b\ell}\PAB(a,b|x,y),
        \end{aligned}
    \end{equation}
    where $\omega:={\rm e}^{\frac{2\pi i}{d}}$ is the $d$-th root of unity and $k,l \in \{0,\cdots,d-1\}$. Effectively, this means we represent Alice's $x$-th measurement outcomes $A_x$ and Bob's $y$-th measurement outcomes by $B_y$ by integer powers of $\omega$.
    
    Starting from Eq. (8) of~\cite{Sarkar:2021aa}, it can be verified that the SATWAP family of Bell inequalities can be cast in a symmetric form via the following steps: (1) simultaneously swap\footnote{This step is unnecessary for symmetrization. However, it allows us to specify the optimal measurements using \cref{Eq:CGLMP-Opt} as it is.} both parties' setting, i.e.,  $x,y: 1 \leftrightarrow 2$, (2) permute Alice's $x$-th measurement outcome so that $a\to (-a-x-1)\!\mod d$. Explicitly, after symmetrization and simplification using $\omega^d=1$, the family of symmetric inequalities reads as:
    \begin{align}\label{Eq:symSATWAP}
            \mathcal{\tilde{I}}_{d} 
            = \sum_{k = 1}^{d-1} 
            \big[
            & \omega^{-2k}\big(a_k\langle A_{1}^{-k}B_{1}^{-k}\rangle +  a_k^* \langle A_{1}^{-k}B_{2}^{-k}\rangle \\
            &+  a_k^*  \langle A_{2}^{-k}B_{1}^{-k}\rangle\big) +  \omega^{-3k}a_k\langle A_{2}^{-k}B_{2}^{-k}\rangle\big]\overset{\L}{\le} \beta_L,\nonumber
    \end{align}
    where $a_k=\frac{(1-i)}{2}\omega^{\frac{k}{4}}$, $(\cdot)^*$ denotes complex conjugation, and $\beta_L = \frac{1}{2}[3\cot{\frac{\pi}{4d}} -\cot{\frac{3\pi}{4d}}]-2$ is the local upper bound.

    Our computation results show that for the symmetric inequalities of~\cref{Eq:symSATWAP}, the measurement of \cref{Eq:CGLMP-Opt}, in conjunction with a symmetric maximally entangled two-qudit state, indeed reproduces the maximal quantum violation~\cite{SAT+17} of $2(d-1)$ for $d\le 100$. Together with the fact~\cite{Sarkar:2021aa} that this maximal quantum violation self-tests a maximally entangled two-qudit state,  we know that the symmetric family of SATWAP inequalities, \cref{Eq:symSATWAP}, also cannot exhibit a trade-off for $d\le 100$. Moreover, we conjecture that this holds true for any higher $d$ as well.

    Incidentally, all the Bell inequalities analyzed above fall under the $(2,2,n)$ Bell scenarios, involving only two measurement settings. In the following, we switch to Bell inequalities defined for the Bell scenarios $(2,m,2)$ with $m\ge 3$, i.e., with more measurement settings, but only binary measurement outcomes. We shall see that this modification allows us to easily identify Bell inequalities where any SQS in the minimal dimension is necessarily suboptimal.
        
\section{Examples where SQS in the minimal dimension must be suboptimal}    
    \label{sec:Asymmetriccase}        
    
    Next, we present some examples of {\em symmetric} Bell inequalities where, for the purpose of achieving their maximal violation, a trade-off between symmetry and dimension has been found.
    
    \subsection{Inequalities in the $(2,3,2)$ Bell scenario}
        \label{Sec:3322}
       
    The complete set of facet-defining Bell inequalities for this Bell scenario was first determined by Froissart~\cite{Froissart1981}, then independently rediscovered by Pitowsky and Svozil~\cite{PS01}, Sliwa~\cite{Sliwa03}, as well as Collins and Gisin~\cite{Collins04}. In this case, apart from the (input-lifted~\cite{Pironio_Lifting_05}) CHSH inequality of~\cref{Eq:CHSH}, there is also the so-called $I_{3322}$ inequality, commonly known in an asymmetric form~\citep[Eq.~(19)]{Collins04}. However, via an appropriate relabeling, we can also rewrite this inequality in a PPI form: see~\cite{Sliwa03} and~\cite{Brunner2008} for a symmetric form of this inequality written, respectively,  in terms of correlators and probabilities, cf.~\cref{eqn:bellineq}. Since $\Dm$ for this inequality could well be infinite~\cite{PV10}, and one can always follow the procedure outlined in~\cref{Fig:TowardPSQS} to obtain an SQS that is still infinite dimensional, there is little hope to exhibit a trade-off between symmetry and dimension for this Bell inequality.

    Since there is no hope of finding a trade-off in this Bell scenario using a facet-defining Bell inequality, naturally, one may wonder whether the same conclusion also holds for non-facet-defining Bell inequalities. At first glance, the symmetric $I_{3322}$-like Bell inequality of \citep[Eq. (27)]{Goh2018} may seem like an example exhibiting a trade-off, since it was shown therein that its maximal violation can be attained using a family of PPI correlations $\Psym$ that arise from {\em asymmetric} two-qubit strategies. However, a closer inspection reveals that these strategies---as with~\cref{eq:max_CHSH_strategy}---can be transformed into a symmetric form via a local unitary, see~\cref{App:I3322LikeSQS} for details. 
    
    In contrast, consider for $\alpha\in[1.5,3]$ the following family of symmetric Bell inequalities in correlator form:
       \begin{align}
       \label{eq:I_S}
           {I}_{S}(\alpha) &=  
           \langle A_{0}\rangle +\langle B_{0}\rangle +\langle A_{1}\rangle +\langle B_{1}\rangle +\alpha\l(\langle A_{2}\rangle +\langle B_{2}\rangle \r ) \nonumber \\
           &\quad+\langle A_{0}B_{2}\rangle+  \langle A_{2}B_{0}\rangle - \langle A_{2}B_{1}\rangle - \langle A_{1}B_{2}\rangle \nonumber\\
            &\quad-\expA{A_2B_2}-2\l (\langle A_{1}B_{0} \rangle + \langle A_{0}B_{1} \rangle + \langle A_{1}B_{1} \rangle \r )  \nonumber\\ 
           & \overset{\L}{\le} 2\alpha+5,
       \end{align}
        where the local bound is easily verified to be saturated by the symmetric local deterministic strategy
       \begin{equation}
           A_x=(-1)^x,\quad B_y=(-1)^y\quad\forall\,\,x,y\in\{0,1,2\}.
       \end{equation}
          
       For concreteness, consider now the special case of $\alpha=2$. Quantum mechanically, it can be verified that an {\em asymmetric} two-qubit strategy (see \cref{App:Is} for details), which involves a degenerate measurement for $\hat{B}_2$, leads to the value $\frac{1}{3}{\left(13+4\sqrt{13}\right)}\simeq 9.1407$. Moreover, this quantum violation arising from the resulting {\em asymmetric} correlation agrees with the quantum upper bound on $I_S(2)$ obtained from the NPA hierarchy to a precision better than $10^{-8}$. Hence, $\Dm=2$ for this inequality. However, if we, instead, only consider SQS facilitated by qubits and PVMs,\footnote{Note that it suffices~\cite{CLeve2004,Liang:PRA:2007} to consider PVMs in maximizing the Bell violation of {\em all} two-outcome Bell inequalities.} then our computation---using the method described in~\cref{App:SQSDBound}---shows that such strategies cannot even violate the Bell inequality of \cref{eq:I_S} at all. In other words, there is a trade-off between symmetry and dimension in getting the maximal Bell violation of $I_S$.
       
       More generally, as we illustrate in~\cref{fig:MaxI_S}, even though an asymmetric {\em qubit} strategy involving a degenerate measurement for $\hat{B}_2$ can lead to the maximal quantum violation of $I_S(\alpha)$ (up to numerical precision) for all $\alpha\in[1.5,3]$, the maximal quantum violation attainable by a symmetric qubit QS is always suboptimal. In fact, for $\alpha \in (1.975, 3]$, these latter strategies can, at best, give the local bound of $2\alpha+5$.
           
       \begin{figure}[!t]
           \centering
           \includegraphics[width=0.95\linewidth]{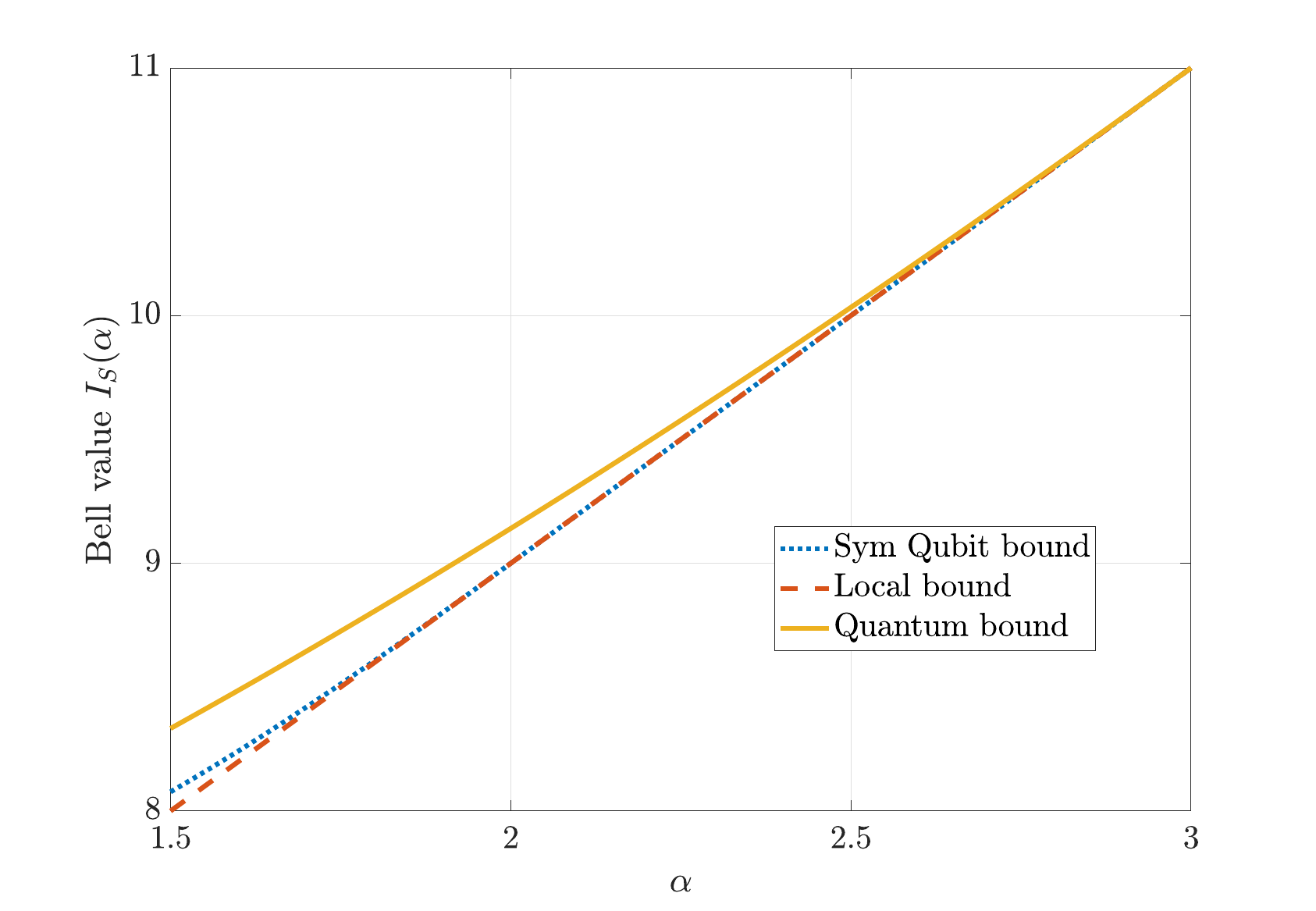}
           \caption{Maximal Bell value of $I_S(\alpha)$ for $\alpha\in[1.5,3]$ under various constraints. From bottom to top, we have, respectively, the local bound of $2\alpha+5$ (red, dashed), the symmetric qubit upper bound (blue, dotted) computed using the method described in \cref{App:Techniques}, and the general quantum bound (yellow, solid), attainable using a two-qubit QS with a degenerate observable for one of the parties.}
           \label{fig:MaxI_S}
       \end{figure}      
    \subsection{Inequalities in the $(2,4,2)$  Bell scenario}
        \label{Sec:242}
    
    \begin{table*}
        \centering
        \begin{tabular}{c|c|cc|c|c|c|ccc}
            &~\bf Quantum bound~ & \multicolumn{2}{c|}{~\bf SQS bound~} &\multicolumn{2}{c|}{}&~\bf Quantum bound~&\multicolumn{3}{c}{~\bf SQS bound ~}\\
            \hline \hline 
            &$\Dm=2$&$d=2$&$d=3$&\multicolumn{2}{c|}{}&$\Dm=3$& $d=3$& $d=4$&$d=5$\\
            \hline
            $I_{4422}^{4}$&0.4142&0.2500&0.3846&\multirow{2}{*}{$I_{4422}^{8}$}&LB&\multirow{2}{*}{0.4878}&\multirow{2}{*}{0.4843}&0.4843 & 0.4843\\
            $I_{4422}^{13}$&0.4349&0.2500&0.3466& &UB& & &0.4878 & 0.4878\\ 
            \cline{5-10} 
            $I_{4422}^{15}$ &0.4349&0.3913 &0.4067&\multirow{2}{*}{$I_{4422}^{19}$}&LB&\multirow{2}{*}{0.4972} &\multirow{2}{*}{0.4514}&0.4832 &0.4836\\
            $J_{4422}^{42}$&0.6722 &0.5682&0.5682&&UB& & &0.4972 &0.4972\\ 
            \cline{5-10}
            $J_{4422}^{86}$&0.7559 &0.7500&0.7500&\multirow{2}{*}{$J_{4422}^{13}$}&LB&\multirow{2}{*}{0.6927} &\multirow{2}{*}{0.6671}&0.6671&0.6722\\
            $J_{4422}^{113}$&0.8484&0.8195&0.8195& &UB& & &0.6927&0.6927\\
            \hline 
        \end{tabular}
        \caption{\label{tbl: asymmetric strategy}
        List of symmetric facet-defining Bell inequalities in the $(2,4,2)$ Bell scenario, where there exists a trade-off between dimension and symmetry in their quantum violation. The first three block columns on the left give, respectively, the list of six inequalities whose $\Dm=2$, their maximal quantum violation (obtained from NPA level 3 and a matching QS of dimension $\Dm$), and their SQS bound for $d\in[\Dm,2\Dm)$. The three block columns on the right give the analogous results for three inequalities whose $\Dm=3$. Note, however, that for these inequalities, the best lower bound (LB) we have found for $d=4, 5$ falls short of the corresponding SDP upper bound (UB).}
    \end{table*}
        
	As we increase the number of settings to four, it is known~\cite{Deza:2016aa,Zambrini19_4422, Oudot_2019} that the Bell polytope is completely characterized by 175 facet-defining classes of Bell inequalities. Among them, 55 are known~\cite{Bancal_2010} to admit a symmetric representation. Several of these, namely, the trivial positivity facet, the CHSH inequality of~\cref{Eq:CHSH}, and the $I_{3322}$ inequality, are liftings~\cite{Pironio_Lifting_05} of facet-defining Bell inequalities from simpler Bell scenarios. 

    For the remaining 52 {\em symmetric} Bell inequalities with four settings, except for $J_{4422}^{62}$ (whose $\Dm$ remains unknown~\cite{PV10}) and 8 others 
    (see~\cref{tbl: asymmetric strategy} as well as \cref{tbl:symmetry class,tbl:HighD-symmetryclass} of~\cref{App:4422}), all the rest can be maximally violated using qubit strategies, i.e., having $\Dm=2$. Moreover, our computation results show that for 42 out of these 52 inequalities, there is no trade-off, i.e., there exists a $\mMSQS$. For ease of reference, we provide this list of 42 inequalities and their maximal quantum violation in~\cref{tbl:symmetry class} of~\cref{App:4422}.

    After discounting $J_{4422}^{62}$, we are left with 9 symmetric inequalities, where we observe a trade-off between dimension and symmetry; see~\cref{tbl: asymmetric strategy} for details. It is worth noting that for all these 9 inequalities, as with $I_S(\alpha)$, the $\mMSQS$ we have found always involves projectors of unequal ranks between Alice and Bob, i.e., $\text{rank}(M_{a|x}^B) \neq \text{rank}(M_{a|x}^A)$ for some $a,x$, thereby making these QSs evidently {\em asymmetric}. In fact, a closer inspection reveals that the resulting correlation is also asymmetric.

    Also worth noting is that for the $J_{4422}^{42}$ inequality, we further observe a gap between the maximal quantum value attainable using qubit SQS and symmetric correlations that may also arise from asymmetric strategies. Indeed, for these two classes of strategies, the maximal quantum values attainable are $0.5682$ and $0.6012$, respectively.
    More explicitly, the latter value can be achieved by adopting the following asymmetric QS consisting of the shared state
    \begin{subequations}\label{eq:J42_strategy}
        \begin{equation}\label{eq:J42_state}
           \ket{\psi} =-\sin{\a} \ket{00}+\cos{\a}\ket{11}, 
        \end{equation}
        with $\alpha = 42.5092^\circ$ and the mirror-symmetric observables 
        \begin{equation}
            \hat{A}_x=\vec{a}_{x}\cdot\vec{\sigma},\quad \hat{B}_y=\vec{b}_{y}\cdot\vec{\sigma}=(\vec{a}_{y}\cdot\vec{\sigma})^*, 
        \end{equation}
        where the measurement (Bloch) vectors $\vec{a}_k:= (\sin\theta_k\cos\phi_k,\sin\theta_k\sin\phi_k,\cos\theta_k)$,
    	\begin{equation}
            \begin{aligned}
                &\theta_0 = 61.9767^\circ,&\quad\phi_0 &= 166.1570^\circ,\\
                &\theta_1 = 0^\circ,&\quad \phi_1 &\text{ arbitrary},\\ 
                &\theta_2 = 54.3423^\circ,&\quad \phi_2 &=  41.5892^\circ,\\
                &\theta_3 =52.2700^\circ,& \quad \phi_3 &= -71.170^\circ.
            \end{aligned}
    	  \end{equation}
    \end{subequations}     
    Equivalently, $\vec{b}_k$ may be obtained from $\vec{a}_k$ by the substitution $\phi_k\to-\phi_k$. For graphical representation of $\vec{a}_k$ and $\vec{b}_k$, see~\cref{fig:J42directions}. 
	  
    \begin{figure}[H]
       \centering
       \includegraphics[width=1\linewidth]{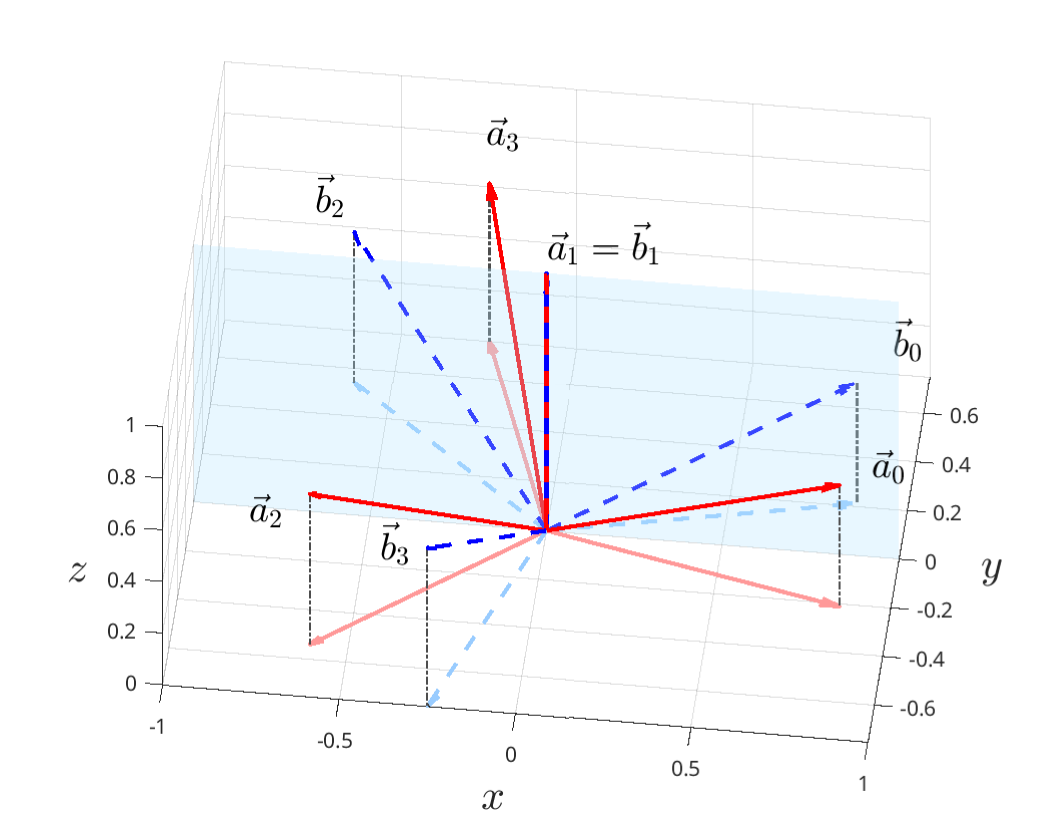}
       \caption{Measurement directions $\vec{a}_k$ (solid line) and $\vec{b}_k$ (dashed line) for $k=0,1,2,3$ corresponding to the QS described in \cref{eq:J42_strategy}. Notice that each $\vec{b}_k$ may be obtained from the corresponding $\vec{a}_k$ by performing a mirror reflection about the $x-z$ plane (the pale blue plane), making it evident that the QS is asymmetric. Surprisingly, the resulting correlation is {\em symmetric} and gives a Bell value $\approx 0.6012$ for the $J^{42}_{4422}$ inequality, higher than the 0.5682 bound achievable by {\em any} qubit SQS.}           
       \label{fig:J42directions}
    \end{figure}

    In the next section, we provide a proof showing that QSs of the kind given in \cref{eq:J42_strategy} not only produce a symmetric correlation, but also cannot be symmetrized via a local unitary transformation.

\section{Asymmetric strategies giving symmetric correlations}
    \label{Sec:Asym2Sym}
        
	\subsection{Symmetric correlations from mirror-symmetric strategies}
        \label{App:MirSym}

    To see that the correlation resulting from \cref{eq:J42_strategy} is indeed PPI, we remark that a more general class of mirror-symmetric strategies must also produce symmetric correlations. Indeed, \cref{eq:J42_state} is easily seen to be a special case of the following family of two-qubit pure states:
    \begin{equation}\label{Eq:State:RealSym_ImgAntiSym}
        \ket{\phi(\alpha)} = \cos{\alpha}\ket{\psi}_{\text{Sym}} + i\sin{\alpha} \ket{\Psi^-},
    \end{equation}
    where $\ket{\Psi^-} = \frac{1}{\sqrt{2}}(\ket{01}-\ket{10})$ is the singlet state and $\ket{\psi}_{\text{Sym}}$ is some real symmetric state (i.e., a real, linear combination of the remaining three Bell states $\ket{\Phi^\pm}$ and $\ket{\Psi^+}$). 
        
    Moreover, if we locally measure~\cref{Eq:State:RealSym_ImgAntiSym} for a pair of measurement directions $\{\vec{a}_k\}, \{\vec{b}_k\}$ mirror-symmetric with respect to the $x-z$ plane of the Bloch sphere, i.e.,
    \begin{equation}\label{Eq:MirrorMeas}
        \mathcal{M} \vec{a}_k = \vec{b}_k ~\forall\,\, k\quad\text{where}\quad \mathcal{M}=\text{diag}(1,-1,1),
    \end{equation}
    it follows that the resulting correlation $\vecP$ must be symmetric.
	\begin{theorem}\label{Thm:SymCor}
        Quantum correlation obtained by performing an arbitrary number of dichotomic measurements that are mirror-symmetric with respect to the $x-z$ plane, c.f. \cref{Eq:MirrorMeas}, on the asymmetric state of \cref{Eq:State:RealSym_ImgAntiSym} is symmetric.
    \end{theorem}
	
    To see that this is the case, we start by establishing the following Lemma.

    \begin{lemma}
        \label{lem:MirrorPOVM}
        If a pair of measurement bases $\vec{a}_k, \vec{b}_k$ are mirror-symmetric with respect to the $x-z$ plane, i.e., \cref{Eq:MirrorMeas} holds, then their POVM elements satisfy 
    	\begin{equation}\label{Eq:Meas:MirSym}
    	    M^B_{a|k} = (M^A_{a|k})^*\quad\forall\,\,a,k.
    	\end{equation}	
    \end{lemma}
    \begin{proof}
        For dichotomic measurements, the POVM elements can be expressed as
        \begin{equation}
            M^A_{a|k} = \frac{1}{2}[\id_2 + (-1)^a \vec{a}_k \cdot \vec{\sigma}], ~ M^B_{b|k} = \frac{1}{2}[\id_2 + (-1)^b \vec{b}_k \cdot \vec{\sigma}],
        \end{equation}
        where $\vec{\sigma}$  is the vector of Pauli matrices, and $a,b \in \{0,1\}$.  
        Since $\mathcal{M} \vec{a}_k = \vec{b}_k$, we have
        \begin{equation}
            \vec{b}_k \cdot \vec{\sigma} = (\mathcal{M} \vec{a}_k) \cdot \vec{\sigma} = a_k^1 \sigma_x - a_k^2 \sigma_y + a_k^3 \sigma_z = (\vec{a}_k \cdot \vec{\sigma})^*.
        \end{equation}
        Therefore, $ M^B_{b|k} = \frac{1}{2}[\id_2 + (-1)^b (\vec{a}_k \cdot \vec{\sigma})^*] = (M^A_{b|k})^*$.
    \end{proof}
    
    The proof of \cref{Thm:SymCor} may now be completed as follows.

    \begin{proof}    
        Using Lemma~\ref{lem:MirrorPOVM}, the joint probability distribution $P_{AB}(a,b|x,y)$ arising from measuring the state $\ket{\phi(\alpha)}$ with the mirror-symmetric measurements reads as: 
        \begin{equation}
            \begin{aligned}
                \PAB(a,b|x,y) 
                &= \bra{\phi(\alpha)} M^A_{a|x} \ten M^B_{b|y} \ket{\phi(\alpha)}\\
                &= \bra{\phi(\alpha)} M^A_{a|x} \otimes (M^A_{b|y})^* \ket{\phi(\alpha)}.
            \end{aligned}
        \end{equation}
        Since $P_{AB}(a,b|x,y)$ is real-valued, we have
        \begin{equation}
            \begin{aligned}
                \PAB(a,b|x,y) 
                &= \PAB(a,b|x,y)^*\\
                &= \bra{\phi(\alpha)^*} (M^A_{a|x})^* \ten M^A_{b|y} \ket{\phi(\alpha)^*},
            \end{aligned}
        \end{equation}
        where $\ket{\phi(\alpha)^*} \coloneqq \cos\alpha\ket{\psi}_{\text{Sym}} - i\sin\alpha \ket{\Psi^-}$, i.e., the complex conjugation of $\ket{\phi(\alpha)}$.
        With the observation that $S\ket{\phi(\alpha)} = \ket{\phi(\alpha)^*}$, where $S$ is the swap operator, we can further reduce $\PAB(a,b|x,y)$ to
        \begin{equation}
            \begin{aligned}
                \PAB(a,b|x,y) 
                &= \bra{\phi(\alpha)} S^\dagger \big[ (M^A_{a|x})^* \ten M^A_{b|y} \big]S \ket{\phi(\alpha)}\\
                &= \bra{\phi(\alpha)} M^A_{b|y} \ten (M^A_{a|x})^* \ket{\phi(\alpha)}\\
                &= \bra{\phi(\alpha)} M^A_{b|y} \ten M^B_{a|x} \ket{\phi(\alpha)}\\ 
                &= \PAB(b,a|y,x),
            \end{aligned}
        \end{equation}
        where Lemma~\ref{lem:MirrorPOVM} is again utilized to arrive at the third equality.
        Hence, $\PAB(a,b|x,y)$ arising from this strategy satisfies the symmetry condition $\PAB(a,b|x,y) = \PAB(b,a|y,x)$ defined in~\cref{def: Symmetric correlation}.
    \end{proof}

    We now show that {\em no} local operations can symmetrize the measurement bases of the two players provided that their measurement directions are each {\em not} coplanar on the Bloch sphere.
    \begin{proposition}
    \label{Prop:Nonplanar}
        Let $\{\vec a_k\}_{k=1}^m \subset \mathbb R^3$ with $m\ge 3$ be non-coplanar, 
        i.e., $\dim(\text{span}\{\vec a_k\})=3$, and let $\{\vec b_k\}_{k=1}^m$ satisfy \cref{Eq:MirrorMeas}, where $\mathcal M\in O(3)$ is a reflection.
        Then there do not exist local unitaries $U_A, U_B\in SU(2)$ such that the resulting measurement directions satisfy $\vec a'_k=\vec b'_k$ for all $k$.
    \end{proposition}
    \begin{proof}
        Suppose the contrary that there exists $U_A,U_B\in SU(2)$ such that 
        \begin{equation}
            \begin{aligned}
                U_A (\vec{a}_k \cdot \vec{\sigma}) U_A^\dagger 
                &= \vec{a}_k' \cdot \vec{\sigma} &&\forall\,\, k,\\
                U_B (\vec{b}_k \cdot \vec{\sigma}) U_B^\dagger 
                &= \vec{b}_k' \cdot \vec{\sigma} &&\forall\,\, k,
            \end{aligned}
        \end{equation}
        with $\vec a'_k=\vec b'_k$ for all $k$.
        
        Let $U \coloneqq U_B^\dagger U_A$. Since $\vec a'_k=\vec b'_k$ for all $k$, we have $U (\vec{a}_k \cdot \vec{\sigma}) U^\dagger = \vec{b}_k \cdot \vec{\sigma}~\forall\,\,k$. So we can focus on proving the existence of this single unitary $U$.
        It is well-known that for each $U\in SU(2)$ there exists a rotation $\mathcal R\in SO(3)$ such that $U(\vec v\cdot\vec\sigma)U^\dagger = (\mathcal R\vec v)\cdot\vec\sigma, \forall\,\, \vec v\in\mathbb R^3$. Hence, there exists some $\mathcal R$ such that $\mathcal R\vec a_k = \vec b_k~\forall\,\, k$.
        By assumption, we also have $\mathcal{M}\vec{a}_k = \vec{b}_k~\forall\,\,k$.
        Altogether, we have $\mathcal{R}^T \mathcal{M} \vec{a}_k = \vec{a}_k~\forall\,\,k$, using the fact that $\mathcal{R}^T = \mathcal{R}^{-1}$ for any orthogonal matrix $\mathcal{R}$.
        Define $\mathcal{Q}:= \mathcal{R}^T \mathcal{M}$.
        Since the set $\{\vec{a}_k\}$ are not coplanar, they span $\mathbb{R}^3$.
        And since $\mathcal{Q} \vec{a}_k = \vec{a}_k$ for all $k$, the multiplicity of the eigenvalue $1$ of $\mathcal{Q}$ must be $3$, whenever the cardinality of the set $\{\vec a_k\}_{k=1}^m$ satisfies $m \geq 3$.
        In other words, $\mathcal{Q} = \mathbb{I}_3$ (in some basis).
        However, this is impossible since $\det(\mathcal{Q}) = \det(\mathcal{R}^T) \det(\mathcal{M}) = \det(\mathcal{M}) = -1$, where the last equality follows from the fact that $\mathcal M$ is a reflection.
    \end{proof}

    In addition, as a special case of the following Proposition, we see that the only pure states that can generate symmetric correlations with mirror-symmetric qubit measurements are exactly those of the form defined in~\cref{Eq:State:RealSym_ImgAntiSym}.
    
    \begin{proposition}
        Let $\Qstr=\{ \ket{\psi}, \{M_{a|x}^A\}, \{M_{b|y}^B\} \}$ be a QS satisfying $M_{a|k}^B = (M_{a|k}^A)^*~\forall~a,k$ and that produces a symmetric correlation.
        Suppose further that the collection of $\{M_{a|x}^A\}_{a,x}$ over all possible $a,x$ spans the entire space that they act on.
        Then, apart from a global phase factor, the state $\ket{\psi}$ must be of the form 
        \begin{equation}\label{Eq:GeneralMirrorState}
        	\ket{\psi} = c_s \ket{\psi}_\text{Sym} + c_a \ket{\psi}_\text{ASym}.
        \end{equation}
        Here, $c_s \ket{\psi}_\text{Sym}$ is a real-valued symmetric (Sym) state and $c_a \ket{\psi}_\text{ASym}$ is a purely imaginary antisymmetric (ASym) state, i.e.,
        \begin{equation}\label{Eq:Def:Sym+ASym}
        	S\ket{\psi}_\text{Sym} = \ket{\psi}_\text{Sym},\quad S\ket{\psi}_\text{ASym} = -\ket{\psi}_\text{ASym},
        \end{equation}	
        where $S$ is the swap operator defined in~\cref{Eq:Swap}.
    \end{proposition}
    \begin{proof}
        By assumption, $\ket{\psi}$ produces a symmetric correlation when measured with POVMs satisfying~\cref {Eq:Meas:MirSym}.
        For convenience, we write $\ket{\psi^*}:=(\ket{\psi})^*$, $\ket{\psi^*}_\text{Sym}:=(\ket{\psi}_\text{Sym})^*$, and $\ket{\psi^*}_\text{ASym}:=(\ket{\psi}_\text{ASym})^*$.
        To prove the Proposition, we employ the three properties: 
        \begin{enumerate}
            \item[(i)] the correlation is symmetric, i.e., $P_\AB(a,b|x,y) = P_\AB(b,a|y,x)$, 
            \item[(ii)] the correlation is real-valued, i.e., $P_\AB(a,b|x,y) = P_\AB(a,b|x,y)^*$, and 
        \item[(iii)] the POVMs satisfy \cref{Eq:Meas:MirSym}.
        \end{enumerate}
        Using (ii) and (iii), we arrive at the following:
        \begin{equation}
            \begin{aligned}
                &P_\AB(a,b|x,y) = \braket{\psi|M_{a|x}^A \ten M_{b|y}^B|\psi}\\
                &= \braket{\psi^*|(M_{a|x}^A)^* \ten (M_{b|y}^B)^*|\psi^*}\\
                &= \braket{\psi^*|(M_{a|x}^A)^* \ten M_{b|y}^A|\psi^*},
            \end{aligned}
        \end{equation}
        while with (i), (iii), and the defining property of the swap operator, we arrive at
        \begin{equation}
            \begin{aligned}
                &P_\AB(a,b|x,y) = \braket{\psi|M_{a|x}^A \ten M_{b|y}^B|\psi}\\
                &= \braket{\psi|M_{b|y}^A \ten M_{a|x}^B|\psi}\\
                &= \braket{\psi|S[(M_{a|x}^A)^* \ten M_{b|y}^A]S^\dagger|\psi}.
            \end{aligned}
        \end{equation}
        Subtracting the terms at the end of the last two lines yields
        \begin{equation}\label{Eq:X}
            \tr \Big[ (M_{a|x}^A)^* \ten M_{b|y}^A X \Big] = 0, \forall\,\, a,b,x,y,
        \end{equation}
        where $X:=\ketbra{\psi^*}{\psi^*} - S^\dagger\ketbra{\psi}{\psi} S$.
    	Now, since $\{M_{a|x}^A\}$ span the entire space that they act on, $\{(M_{a|x}^A)^* \ten M_{b|y}^A\}$ also span the entire composite Hilbert space.
        In particular, we can consider an appropriate linear combination of these operators, and hence of \cref{Eq:X}, to arrive at $\tr \Big( X^\dag X \Big) = 0$.
        However, the left-hand side of the last equation is simply the Hilbert-Schmidt norm of the operator $X$.
        This means that $X$ must be identically zero, and hence
        \begin{equation}\label{Eq:SPsi}
        	S \ket{\psi} = e^{i\theta} \ket{\psi^*}
        \end{equation} for some $\theta \in [0, 2\pi)$.
        
        Without loss of generality, we can express $\ket{\psi}$ as a linear combination of Sym and ASym components, cf.~\cref{Eq:GeneralMirrorState,Eq:Def:Sym+ASym}
        where 
    	$c_s, c_a \in \C$ while both $\ket{\psi}_\text{Sym}$ and $\ket{\psi}_\text{ASym}$ may be complex.
        Substituting \cref{Eq:GeneralMirrorState} into \cref{Eq:SPsi} 
        gives 
        \begin{equation}\label{Eq:c}
            c_s\ket{\psi}_\text{Sym} = e^{i\theta} c_s^*\ket{\psi^*}_\text{Sym},\quad -c_a\ket{\psi}_\text{ASym} = e^{i\theta} c_a^*\ket{\psi^*}_\text{ASym}, 
        \end{equation}
        where we have used the complex conjugation of~\cref{Eq:Def:Sym+ASym}.
        Let us define $\tilde{c}_s \coloneqq e^{-i\theta/2} c_s$ and $\tilde{c}_a \coloneqq e^{-i\theta/2} c_a$.
        Then, it follows from \cref{Eq:c} that 
        \begin{equation}
        	\tilde{c}_s\ket{\psi}_\text{Sym} = \tilde{c}_s^*\ket{\psi^*}_\text{Sym},\quad -\tilde{c}_a\ket{\psi}_\text{ASym} = \tilde{c}_a^*\ket{\psi^*}_\text{ASym}, 
        \end{equation}	
        indicating that $\tilde{c}_s\ket{\psi}_\text{Sym}$ is a real-valued symmetric vector, while $\tilde{c}_a\ket{\psi}_\text{ASym}$ is a purely imaginary anti-symmetric vector.
        Hence, we can write $\ket{\psi} = e^{i\theta/2} \l( \tilde{c}_s \ket{\psi}_\text{Sym} + \tilde{c}_a \ket{\psi}_\text{ASym} \r)$, which completes the proof. This is equivalent to the form defined in~\cref{Eq:State:RealSym_ImgAntiSym} up to some global phase.
    \end{proof}

    \subsection{More general construction of QS giving symmetric correlations}
    
    The discussion in~\cref{App:MirSym} makes evident that symmetric correlations may arise beyond SQSs. In fact, we may consider both SQSs and mirror-symmetric strategies discussed in \cref{App:MirSym} as special cases of a more general construction. To this end, let us note from \cref{Eq:Born} and \cref{def: Symmetric correlation} that for a symmetric correlation $\vecP=\Psym$, we have
        \begin{equation}\label{Eq:Born2}
            P_\AB(a,b|x,y) = \tr(\rho_\AB M^A_{a|x}\otimes M^B_{b|y})= P_\AB(b,a|y,x).
         \end{equation}
    By definition, we may also write the last term of \cref{Eq:Born2} as
        \begin{equation}\label{Eq:Born3}
            P_\AB(b,a|y,x)=P_\BA(a,b|x,y) = \tr(\rho_\BA M^B_{a|x}\otimes M^A_{b|y}).
         \end{equation}
    Consider now a linear positive map $O:\mathcal{B}(\mathcal{H} \otimes \mathcal{H}) \rightarrow \mathcal{B}(\mathcal{H} \otimes \mathcal{H})$ that leaves the trace in \cref{Eq:Born2} invariant, i.e.,  
        \begin{equation}\label{Eq:BornModified}
            P_\AB(a,b|x,y) = \tr\left[O(\rho_\AB) O(M^A_{a|x}\otimes M^B_{b|y})\right],
         \end{equation}
    for some $\Qstr=\{\rho_\AB,\{M_{a| x}^{A}\}_{a,x},\{M_{b|y}^{B}\}_{b,y}\}$. After substituting \cref{Eq:Born3,Eq:BornModified} into \cref{Eq:Born2}, we see that for the latter to hold, it suffices that the following conditions hold
    \begin{subequations}\label{Eq:SufficientConditions}
    \begin{gather}
    	O(\rho_\AB)=\rho_\BA,\\
    	O(M^A_{a|x}\otimes M^B_{b|y})=M^B_{a|x}\otimes M^A_{b|y}\,\,\forall\,\,a,b,x,y.
    \end{gather}
    \end{subequations}
    
    For example, if we take $O$ to be the identity map, then \cref{Eq:SufficientConditions} implies that:
    \begin{subequations}\label{Eq:SufficientConditions:SQS}
        \begin{gather}
        	\rho_\AB=\rho_\BA,\label{Eq:State:Sufficient} \\
        	M^A_{a|x}\otimes M^B_{b|y}=M^B_{a|x}\otimes M^A_{b|y}\,\,\forall\,\,a,b,x,y.\label{Eq:POVM:Sufficient}
        \end{gather}
    \end{subequations}
    By summing both sides of \cref{Eq:POVM:Sufficient}, e.g., over $b$ and using the normalization condition of a POVM, we see that \cref{Eq:SufficientConditions:SQS} is exactly the condition of an SQS, see~\cref{def:Symmetric strategy}.
    
    On the other hand, if we take $O$ to be the complex conjugation (or equivalently, transposition) operation, then \cref{Eq:SufficientConditions} becomes
    \begin{subequations}\label{Eq:SufficientConditions:Mirror}
        \begin{gather}
        	(\rho_\AB)^*=\rho_\BA,\\
        	(M^A_{a|x})^*\otimes (M^B_{b|y})^*=M^B_{a|x}\otimes M^A_{b|y}\,\,\forall\,\,a,b,x,y,
        \end{gather}
    \end{subequations}
    which is easily verified to hold for mirror-symmetric strategies discussed in \cref{App:MirSym}, see~\cref{Eq:Meas:MirSym,Eq:SPsi}.
    
    For an explicit example, if we take mirror-symmetric qubit measurements, then for any two-qubit state, we have
    \begin{equation}
        \rho
        = \frac{1}{4}\Big(
        \mathbb{I}\otimes\mathbb{I}
        + (\vec r \cdot \vec\sigma) \otimes \mathbb{I}
        + \mathbb{I} \otimes (\vec s \cdot \vec\sigma)
        + \sum_{i,j=1}^{3} T_{ij}\, \sigma_i \otimes \sigma_j
        \Big),
    \end{equation}
    where $\vec r = (r_x,r_y,r_z)$, $\vec s = (s_x,s_y,s_z)$ and $T$ is the correlation matrix. To satisfy~\cref{Eq:SufficientConditions:Mirror}, we require
    $\vec s = (r_x, -r_y, r_z)$ and
    \begin{equation}
        T = 
        \begin{pmatrix}
            T_{xx} & T_{xy} & T_{xz} \\
            -T_{xy} & T_{yy} & T_{yz} \\
            T_{xz} & -T_{yz} & T_{zz} 
        \end{pmatrix}.
    \end{equation}
    This gives the most general form of a two-qubit state that produces a symmetric correlation with mirror-symmetric measurements.
        
\section{Asymmetry and the geometry of the quantum set of correlations}
    \label{Sec:GeometryandSymmetry}
       
    \subsection{Maximal quantum violation of a symmetric Bell inequality by an asymmetric correlation}
                   
    As illustrated above, there exist symmetric Bell inequalities that can be maximally violated by an  {\em asymmetric} correlation. It turns out that this has a nontrivial implication on the geometry of the quantum set of correlations, as we summarize in the following proposition.
    
	\begin{proposition}\label{Prop:AsymPImpliesFlatness}
		If an asymmetric quantum correlation maximally violates a symmetric Bell inequality, then the set of quantum maximizers of this Bell inequality forms at least a one-dimensional flat region (boundary) in the space of correlations.
	\end{proposition}
	\begin{proof}
	The proof is very analogous to that of Proposition~\ref{Prop:SufficiencySymP}. However, instead of \cref{Eq:SymmetrizingP}, consider now	
        \begin{equation}
        \label{eq:symmetriccombination}
            \vecP_c = c\vecP_\AB^\star +(1-c)\vecP_\BA, \ c \in [0,1],
        \end{equation}
    	which is also a maximizer of $I$. Since $\Q$ is convex and $I$ is linear in $\vecP$, only those correlations $\vecP$ lying on the boundary of $\Q$ can be a maximizer of $I$ in $\Q$. In other words, the set of correlations defined by \cref{eq:symmetriccombination} forms a one-dimensional flat region of the quantum boundary, thus concluding the proof. 
	
	Note that the proof above can be straightforwardly generalized to an $N$-partite Bell scenario using the more general notations of $V_\sigma$ and $\vecP$, instead of $V_\AB$ and $\vecP_\AB$. Here, $V_\sigma$ is the permutation operator realizing the permutation $\sigma$ in the entries of $\vecP$, where $\sigma$ is an element of the symmetric group $S_N$ such that $\vecP_\sigma:=V_\sigma\vecP^\star\neq\vecP^\star$.
	
	\end{proof}
	An important consequence of Proposition~\ref{Prop:AsymPImpliesFlatness} is the following Corollary.
	\begin{corollary}\label{Prop:AsymPImpliesNoSelfTest}
		A symmetric Bell inequality $I$ admitting an asymmetric maximizer cannot be used to self-test any reference quantum strategy $\Qstr$ based on the observed maximal violation of $I$ alone.
	\end{corollary}
	\begin{proof}
        Self-testing~\cite{Supic19} of a reference strategy $\Qstr$ based on an observed Bell value refers to the possibility of deducing from this observation that the underlying strategy, modulo irrelevant local degrees of freedom, must be $\Qstr$. A necessary condition for this possibility is that there is only a unique quantum correlation leading to this Bell value. By assumption, the symmetric inequality $I$ admits an asymmetric maximizer, thus it follows from Proposition~\ref{Prop:AsymPImpliesFlatness} that its maximizer is not unique, thus rendering it impossible to perform self-testing from the maximal violation of $I$ alone.
	\end{proof}
	
	Notice, however, that the above Corollary does not preclude the possibility of self-testing the underlying state (but not the measurements) from the observed Bell value, like the examples discussed in~\cite{Jeba:PRR:2019,Kaniewski:PRR:2020,Gigena:PRA:2022}. Also noteworthy is that flat regions of the quantum boundary have been extensively discussed in~\cite{Goh2018}  (see also~\cite{Barizien2024}). However, most of the examples presented therein concern flat regions that are described by a trivial Bell inequality, i.e., one that is not violated by quantum theory. In contrast, the flat region corresponding to $I_S(\alpha)$  of \cref{eq:I_S}, as with that for the inequalities listed in~\cref{tbl: asymmetric strategy}, does not intersect with the Bell polytope.
	
    \subsection{Maximal quantum violation of an asymmetric Bell inequality by a symmetric correlation}
	
	While a symmetric Bell inequality can be maximally violated by a symmetric correlation, we generally cannot hope that the same conclusion holds for an asymmetric Bell inequality. To see this, consider, e.g., the CHSH inequality obtained from \cref{Eq:CHSH2} by the relabeling $B_0\leftrightarrow B_1$:
	\begin{equation}\label{Eq:CHSH-asym}
               \ICHSH =\expA{A_0 B_0} + \expA{A_0 B_1} - \expA{A_1 B_0} + \expA{A_1 B_1} \overset{\L}{\le} 2.
    \end{equation}
	If a correlation $\vecP$ violates this inequality, one would have 
	\begin{equation}\label{Eq:CHSH-asym-violate}
        \expA{A_0 B_0} + \expA{A_0 B_1} - \expA{A_1 B_0} + \expA{A_1 B_1} > 2.
    \end{equation}
    If the $\vecP$ is also PPI, then the correlation remains invariant under the transformation $A_k\leftrightarrow B_k$ for all $k=0,1$, which means that it must also satisfy, from~\cref{Eq:CHSH-asym-violate}, after a simultaneous relabeling for $\vecP$ and $\vec{\beta}$:
	\begin{equation}\label{Eq:CHSH-asym-violate2}
        \expA{A_0 B_0} + \expA{A_1 B_0} - \expA{A_0 B_1} + \expA{A_1 B_1} > 2.
    \end{equation}
	However, the conjunction of \cref{Eq:CHSH-asym-violate,Eq:CHSH-asym-violate2} contradicts the fact~\cite{LHB+10} that any correlation $\vecP$ can violate at most one version of the CHSH Bell inequality. In other words, no symmetric correlation can violate the asymmetric CHSH Bell inequality of~\cref{Eq:CHSH-asym-violate}, let alone maximally. 
	
	In contrast, consider the following two-parameter family of Bell inequalities:
	\begin{subequations}\label{Eq:TsirelsonIneqInspired}
        \begin{align}
	        I_{r_0,r_1} =\,\, &r_0\left[\frac{A_0+A_1}{\sqrt{2}}-B_0\right]+r_1\left[\frac{A_0-A_1}{\sqrt{2}}-B_1\right]\nonumber\\
	        &+\frac{1}{2\sqrt{2}}\beta_\text{CHSH}\overset{\L}{\le} g(r_0, r_1)
        \end{align}
        for
    	\begin{equation}
    		r_1\neq-(\sqrt{2}+1)r_0,	
    	\end{equation}
	\end{subequations}
	and the pair $(r_0,r_1)$ could otherwise take {\em any} value within the {\em interior} of the octagon spanned by the 8 vertices:
    \begin{equation}\label{Eq:Vertices}
        \pm(\frac{1}{2}-\zeta,0),\,\pm(\zeta,\zeta),\, \pm(0,\frac{1}{2}-\zeta),\, \pm(\zeta,-\zeta),\,
    \end{equation}
	with $\zeta:= \frac{1}{\sqrt{2}}-\frac{1}{2}$. For completeness, we provide the local bound $g(r_0,r_1)$ of this two-parameter family of inequalities in \cref{Eq:g} of \cref{App:LocalBound}.
	It is easy to verify that the Bell inequalities of \cref{Eq:TsirelsonIneqInspired} cannot be written in a PPI form, even after any relabeling of the settings and outcomes. At the same time, it has been shown in~\cite{Barizien2024} that they are all maximally violated by the Tsirelson correlation $\vecP_T$ of \cref{eq:Tsirelson pt}, which is PPI. In other words, the family of asymmetric Bell inequalities given by \cref{Eq:TsirelsonIneqInspired} are all maximally violated by the same symmetric correlation $\vecP_T$.

\section{Conclusion}
    \label{Sec:Conclusion}

Due to the arbitrariness in the classical labeling (of settings, outcomes, and parties) as well as the degeneracy arising from the normalization and no-signaling constraints, a Bell inequality may be rewritten in various forms. However, even after incorporating these degrees of freedom, {\em not all} Bell inequalities can be cast in a symmetric PPI form. In this work, we have focused on (facet-defining) Bell inequalities that indeed admit a symmetric representation. 
	
In particular, we have explored the possibility of finding, for symmetric Bell inequalities, a symmetric quantum strategy (SQS) of {\em minimal} dimension that maximizes their Bell violation. Our results show that, for integer $d\in[2,19]$, there exists an SQS in dimension $d$ maximizing the Bell violation of the family of symmetric $I_{22dd}$ inequalities (known~\cite{Liang:PhDthesis} to be equivalent to the CGLMP Bell inequalities). Moreover, for $d\le5$, no QS of a smaller Hilbert dimension can attain the same maximal value.

However, for a family of symmetric Bell inequalities in the $(2,3,2)$ Bell scenario and 9 of the nontrivial, symmetric, genuine four-setting facet-defining Bell inequalities applicable to the $(2,4,2)$ Bell scenario, we observe a gap between the maximal quantum violation achievable and that arising from an SQS of minimal dimension. In other words, there exists a {\em trade-off} between symmetry and dimension for achieving the maximal quantum violation of these inequalities---either we opt for an {\em SQS} that is HD, or we go for an asymmetric QS that is of {\em minimal dimension}, but not both. Note, however, that we do not know if one can find an example of such a trade-off in a simpler Bell scenario, a problem that may deserve further investigation.

Interestingly, since the optimal, {\em minimal}-dimension QS for all these symmetric inequalities gives rise to an {\em asymmetric} correlation, it follows from Proposition~\ref{Prop:AsymPImpliesFlatness} that the quantum maximizers of each of these Bell inequalities must define a {\em flat} region of the quantum boundary. As explained in~\cref{Sec:GeometryandSymmetry}, this feature has negative implications on the possibility of performing self-testing using the observed quantum violation of these inequalities. 
	
Also worth noting is that among the inequalities showing a trade-off, the maximal (qubit) violations of $J_{4422}^{42}$ under different constraints satisfy:
\begin{equation}
     \max_\text{Qubit SQS} J_{4422}^{42} <\max_\text{Qubit $\Psym$} J_{4422}^{42} <\max_\text{Qubit $\vecP$} J_{4422}^{42} = \max_{\vecP\in\Q} J_{4422}^{42}.
\end{equation}
Hence, with the assumption that the employed QS is a qubit strategy, one may certify that the underlying strategy is asymmetrical, even if the observed correlation is symmetric, i.e., 
\begin{equation}
      J_{4422}^{42} \overset{\text{Qubit SQS}}{\le}  0.5682
\end{equation}
serves as a semi-DI~\cite{Liang2011} witness for an asymmetric qubit QS. In fact, the trade-off is so strong that even if we allow arbitrary symmetric correlations attainable with two-qubit QSs, there remains a gap with the maximal quantum violation of $J_{4422}^{42}$. In contrast, we present in~\cref{App:I9} an example of a $9$-setting Bell inequality that again shows a gap between the maximal two-qubit SQS violation and that achievable with two-qubit symmetric correlations, but where the latter already achieves the quantum maximum. It is worth noting that these witnesses may alternatively be interpreted as a special kind of {\em dimension} witnesses~\cite{Brunner08,Moroder13} (see also~\cite{Lim2010,KST+19}), that are only applicable to SQSs.

\begin{table}
    \centering
    \begin{tabular}{c|c|c|c|c||c}
        Bell Inequality&Facet&$\BIsym$&$\Dm$&$\mMSQS$&(SQS, $\Psym$, $\vecP^\star$)\\ 
        \hline
        CHSH~\cite{Clauser69}:~\cref{Eq:CHSH2}&\cmark&\cmark&2&\cmark&\cref{eq:max_CHSH_strategy}:(\xmark, \cmark, \cmark)\\
        CHSH~\cite{Clauser69}:~\cref{Eq:CHSH2}&\cmark&\cmark&2&\cmark&\cref{eq:CHSH ss strategy}:(\cmark, \cmark, \cmark)\\
        $I_{2233}$~\cite{Collins04}:~\cref{eq: symmetry 22dd}&\cmark&\cmark&3&\cmark&\cref{Eq:CGLMP-Opt}:(\cmark, \cmark, \cmark)\\ 
        \hline
        $I_{3322c}$~\cite{Goh2018}:~\cref{Eq:CorI3322}&\xmark&\cmark&2&\cmark&\cref{Eq:I3322cStrategy}:(\cmark, \cmark, \cmark)\\
        $I_{S}(\alpha)$:~\cref{eq:I_S}&\xmark&\cmark&2&\xmark&\cref{eq:OptimalQS:Is}:(\xmark, \xmark, \cmark)\\
        \hline
        $J_{4422}^{42}$\cite{Pal09QB}:~\cref{Ineq:J42}&\cmark&\cmark&2&\xmark&\cref{eq:J42_strategy}:(\xmark, \cmark, \xmark)\\
        \hline
        $I_{r_0,r_1}$~\cite{Barizien2024}:~\cref{Eq:TsirelsonIneqInspired}&\xmark&\xmark&2&\cmark&\cref{eq:CHSH ss strategy}:(\cmark, \cmark, \cmark)\\
    \end{tabular}
    \caption{\label{tab:summary} Summary of some of the examples presented in this work. From left to right, we list the Bell inequality considered (together with the equation number where we list its explicit form), its facet-defining property, its symmetric property (PPI or not), the dimension $\Dm$ of its minimal maximizing quantum strategy $\mMQS$, and whether there exists a $\mMSQS$, i.e., a $\mMQS$ that is also symmetric (its absence, marked by \xmark, indicates the existence of a trade-off between symmetry and dimension). Further to the right, we list the equation number of an explicit QS considered, and indicate whether it is an SQS, whether the resulting correlation is symmetric ($\Psym$), and whether it maximizes the corresponding Bell-inequality violation ($\vecP^\star$).}
\end{table}

Of course, by considering a nontrivial, asymmetric Bell inequality, one may also hope to drop the dimension assumption and go for a fully DI witness for asymmetry. Indeed, for {\em all} facet-defining Bell inequalities in the $(2,m,2)$ and $(2,2,n)$ Bell scenarios with $m\le 4$ and $n\le 5$, we observe that there always exists an {\em asymmetric} representation that {\em cannot} be violated by symmetric quantum correlations. It will certainly be interesting to see if this feature holds true in general, specifically, for all Bell scenarios. For a summary of some other results we have obtained, see~\cref{tab:summary}. 

An observant reader may have also noticed that we have obtained results of rather different natures in the $(2,m,2)$ and $(2,2,n)$ Bell scenarios---trade-offs in many of the former but none in the latter. Naturally, it will be interesting to investigate what happens in more complicated Bell scenarios $(N,m,n)$ where $N\ge 2$ and both $m,n\ge 3$. Some natural candidates to consider include the symmetric Bell inequalities discussed in~\cite{Liang09,Lim2010,Bancal_2010,Bancal12,KST+19}, as well as inequalities from~\cite{BKP06,Schwarz16,Deza:2016aa,SAT+17,Cope19} that may be cast in a symmetric form. Throughout, we have relied on SDP tools described in~\cref{App:Techniques} to determine whether a trade-off can be demonstrated for any given Bell inequality. 
Evidently, it will be desirable to devise analytic criteria to determine whether such a trade-off can exist before running any numerical optimizations, a challenging problem we shall leave for future consideration.

Note also that the realizability of symmetric distributions by a symmetric model has also been explored in the context of network nonlocality. In fact, it has recently been shown~\cite{WRB+2025} that there exist symmetric distributions in the triangle network (under exchange of parties) that cannot be created via a symmetric model assuming {\em only} no-signaling, regardless of the dimension or nature of the system. Moreover, such intriguing distributions can be created via a local asymmetric model. 

Finally, note also that while we have focused on the symmetry of (full) permutation invariance---see~\cref{App:Multipartite} for some further results in this context---there are clearly interesting alternatives to consider. For instance, one may consider Bell inequalities that exhibit translationally (also called cyclic-permutation) invariance \cite{Grandjean2012,Tura_2014,Wang17}. Could one find a similar kind of trade-off between dimension and symmetry? We leave the exploration of this and other interesting questions for future work.

\begin{acknowledgments}
    We thank Jean-Daniel Bancal, Istv\'an M\'arton, K\'aroly F. P\'al, Gilles P\"utz, William Slofstra, and Elie Wolfe for helpful discussions. HYH thanks Bo-An Tsai for sharing his code, which led to some of the results presented here. KSC thanks the hospitality of the Institut Néel, where part of this work was completed.
    MEL acknowledges support from the Ministry of Education of Taiwan and Université Paris-Saclay through the Taiwan–Université Paris-Saclay Scholarship Program.
    This work was supported by the National Science and Technology Council, Taiwan (Grants No. 109-2112-M-006-010-MY3, 112-2628-M-006-007-MY4, 113-2917-I-006-023, 113-2918-I-006-001), the Foxconn Research Institute, Taipei, Taiwan, in part by the Higher Education Sprout Project, Ministry of Education, to the Headquarters of University Advancement at the National Cheng Kung University (NCKU), and in part by the Perimeter Institute for Theoretical Physics. Research at Perimeter Institute is supported by the Government of Canada through the Department of Innovation, Science, and Economic Development, and by the Province of Ontario through the Ministry of Colleges and Universities. T.~V.~acknowledges the support of the European Union (CHIST-ERA MoDIC), the National Research, Development and Innovation Office NKFIH (Grants No.~2023-1.2.1-ERA\_NET-2023-00009 and No.~K145927) and the `Frontline' Research Excellence Program of the NKFIH (No.~KKP133827). 
\end{acknowledgments}
        
\appendix
        
\crefname{appendix}{appendix}{appendices}
\Crefname{appendix}{Appendix}{Appendices}
\crefalias{section}{appendix} 
\crefalias{subsection}{appendix} 
\crefalias{subsubsection}{appendix} 
        
\section{Bounding quantum values of Bell inequality allowed by symmetric correlations and symmetric strategies}
    \label{App:Techniques}

	Let $\mathcal{B}(\mathcal{H})$ be the set of operators acting on $\mathcal{H}$.
	Given a bipartite Bell inequality specified by the coefficients $\vec{\beta}$, our goal is to obtain its maximal value over all possible symmetric quantum strategies (SQSs) of local HSD upper bounded by $d\le D$, i.e., 
	
    \begin{align}\label{Eq:Opt:SQS}
        \max_{\rho_\AB, \{M_{a|x}\} } \quad &\qquad\  \vec{\beta} \cdot \vecP\nonumber\\
        \text{such that}\quad & \ \PAB(a,b| x,y)=\mathrm{tr}(\rho_\AB M_{a|x}\otimes M_{b|y}),\nonumber\\
        &\ \tr(\rho_\AB)=1,\quad \rho_\AB = S\rho_\AB S^{\dagger} \succeq 0,\nonumber\\
        &\ \rho_\AB\in\mathcal{B}(\mathbb{C}^D\otimes\mathbb{C}^D), \\
        &\ \sum_{a} M_{a|x} = \id_D\ \forall \ x,\quad M_{a|x}\succeq 0,\ \forall \ a,x, \nonumber
    \end{align}
	where $S$ defined in \cref{Eq:Swap} is the swap operator.
	Since an SQS necessarily gives a symmetric correlation, see~\cref{eqn: symmetric strategy}, this problem may be relaxed, e.g., by optimizing over symmetric correlations attainable by QSs of the same HSD, i.e.,
	\begin{align}\label{Eq:Opt:SymmCorr}
        \max_{\rho_\AB, \{M^A_{a|x}\}, \{M^B_{b|y}\} } \,\, &\qquad\  \vec{\beta} \cdot \vecP\nonumber\\ 
        \text{such that}\qquad & \ \PAB(a,b| x,y)=\PAB(b,a|y,x),\nonumber\\
        &\ \PAB(a,b| x,y)=\mathrm{tr}(\rho_\AB M^A_{a|x}\otimes M^B_{b|y}),\nonumber\\
        &\ \tr(\rho_\AB)=1,\,\, \rho_\AB \succeq 0,\,\, \nonumber\\
        &\ \rho_\AB\in\mathcal{B}(\mathbb{C}^D\otimes\mathbb{C}^D), \\
        &\ \sum_{a} M^A_{a|x} = \id_D\ \forall \ x,\quad M^A_{a|x}\succeq 0,\ \forall \ a,x, \nonumber\\
        &\ \sum_{b} M^B_{b|y} = \id_D\ \forall \ y,\quad M^B_{b|y}\succeq 0,\ \forall \ b,y. \nonumber
    \end{align}
	In other words, the optimum value of \cref{Eq:Opt:SymmCorr} is always larger than or equal to that of \cref{Eq:Opt:SQS}.		
		
    \subsection{Upper bounding the quantum violation of a Bell inequality}
        \subsubsection{Without dimensional constraints}
       
        When there is no dimension constraint, i.e., if $D=\infty$, there is no distinction between the Bell value achievable by SQSs and symmetric correlations, see Proposition~\ref{Prop:SCor2SQS}. Then, an upper bound on the maximal quantum violation by SQS can be obtained by solving any SDP hierarchy that outer approximates the quantum set $\Q$, such as that due to Navascu\'es-Pironio-Ac\'in (NPA)~\cite{NPA,NPA2008} and Moroder {\em et al.}~\cite{Moroder13}, but now with the {\em additional constraints} that all the moments are invariant under the action of the symmetry group $G$. For example, in the case of a bipartite Bell scenario, see \cref{def: Symmetric correlation}, we require that all moments remain invariant under the exchange $A\leftrightarrow B$, which means that all moments of the kind $\expA{M^A_{a_1|x_1}\cdots M^A_{a_k|x_k}M^B_{b_1|y_1}\cdots M^B_{b_\ell|y_\ell}}$ can be taken to be the same as $\expA{M^A_{b_1|y_1}\cdots M^A_{b_\ell|y_\ell}M^B_{a_1|x_1}\cdots M^B_{a_k|x_k}}$. To this end, one may first take the moment matrix $\Gamma$ at any given level of these hierarchies and perform the twirling,
        \begin{equation}\label{Eq:Twirling}
            \bar{\Gamma} = \frac{1}{|G|}\sum_{\sigma\in G}U_\sigma\Gamma U_\sigma^\dag,
        \end{equation}
        where $U_\sigma$ is the unitary operator that effects the appropriate row-permutation of $\Gamma$ corresponding to the action of $\sigma$. Clearly, $\bar{\Gamma}$ consists of fewer independent moments compared with $\Gamma$. A symmetric upper bound is then obtained by maximizing the Bell value subject to $\bar{\Gamma}\succeq 0$ instead of $\Gamma\succeq 0$. For further information, see the detailed discussion given in~\cite{rosset2018symdpoly,Ioannou:2021qmh}.
        
        \subsubsection{With dimensional constraints}\label{App:SQSDBound}

        When a local dimension bound $d\le D$ is present, a powerful analog of the NPA hierarchy was proposed in~\cite{DNPA15,NFA15}. In particular, the software package $\texttt{QDimSum}$ discussed in \cite{TRR19} has been developed to implement precisely this and other related finite-dimensional optimizations. To this end, notice that whenever a symmetric Bell inequality is provided to $\texttt{QDimSum}$ as an objective function, \cref{Eq:Twirling} is implemented with respect to the symmetric group $G$ that leaves the Bell inequality invariant. Hence, whenever we use $\texttt{QDimSum}$ to maximize the value of a PPI Bell inequality, we naturally obtain an upper bound on~\cref{Eq:Opt:SymmCorr}, i.e., the maximal value allowed by symmetric correlations arising from dimension-bounded (not necessarily symmetric) QSs.
        
        To obtain a (tighter) upper bound applicable to SQS, i.e.,~\cref{Eq:Opt:SQS}, we modify the sampling procedure described in~\cite{DNPA15,NFA15} for generating random moment and localizing matrices. Specifically, instead of sampling normalized random pure states from $\mathcal{H}_A\otimes\mathcal{H}_B= \mathbb{C}^{D}\otimes\mathbb{C}^{D}$, we randomly sample symmetric (antisymmetric) pure states from $\mathcal{H}_A\otimes\mathcal{H}_B$ by considering random, but normalized linear combination of vectors from the symmetric (antisymmetric) subspace of $\mathcal{H}_A\otimes\mathcal{H}_B$. Moreover, instead of sampling random POVMs for Alice and Bob independently, we randomly sample Alice's POVM and make Bob's POVMs identical to those of Alice.

    \subsection{Lower bounding the quantum violation of a Bell inequality}
        \subsubsection{Over symmetric quantum strategies}
            \label{meth:fminunc}

        	To obtain a lower bound for the maximization problem of~\cref{Eq:Opt:SQS}, we may:
        	\begin{enumerate}
            	\item Without loss of generality, take $\rho_\text{AB}=\proj{\psi}$ to be a pure state and choose $\ket{\psi}$ to be an arbitrary, normalized linear combination of basis states in the (anti)symmetric subspace.
            	\item Set each $M_{a|x}$ to be a projector 
                	\begin{equation}
                		M_{a|x}=\Pi_{a|x}=\Pi_{a|x}^2=\sum_{j=1}^{r_{a|x}} \proj{\phi_{a,x,j}}
                	\end{equation}
                	such that the sum of their rank $\sum_a r_{a|x} = D$ for all $x$, where $\{\ket{\phi_{a,x,j}}\}_{a,j}$ may be taken as an orthonormal set of column vectors forming a $D\times D$ unitary matrix $U_x$.
            	\item Optimize over the the parameters defining $\ket{\psi}$ and the unitary operators $\{U_x\}_x$.
            \end{enumerate}
            Notice that for step (1), a general symmetric and antisymmetric state can be expressed, respectively, as a linear combination of $\frac{D(D+1)}{2}$ symmetric and $\frac{D(D-1)}{2}$ antisymmetric {\em basis} states. Moreover, for step (2), the $\texttt{QLib}$ package \cite{Machnes2007qlib} provides a convenient parametrization of any element of SU($D$) via $D^2-1$ real parameters. All these parameters can then be optimized over, e.g., using the $\mathsf{fminunc}$ function in MATLAB to obtain a local maximum of \cref{Eq:Opt:SQS}. Since these are high-dimensional optimizations over $O((m+1)D^2)$ real parameters, with $m$ being the number of measurement settings, it is usually necessary to perform multiple optimizations to get a good lower bound.

        \subsubsection{Over symmetric quantum correlations}\label{meth:seesaw}
        
        For the maximization problem of~\cref{Eq:Opt:SymmCorr}, which provides an upper bound to that of~\cref{Eq:Opt:SQS}, we may adapt the (see-saw) algorithm described in~\cite{Liang:PRA:2007,Pal09QB,Liang09} to iteratively optimize over
        \begin{enumerate}
        \item Alice's POVM $\{M^A_{a|x}\}_{a,x}$ for given state $\rho_\text{AB}$ and Bob's POVM $\{M^B_{b|y}\}_{b,y}$;
        \item\label{OptOverB} Bob's POVM $\{M^B_{b|y}\}_{b,y}$ for given state $\rho_\text{AB}$ and Alice's POVM $\{M^A_{a|x}\}_{a,x}$;
        \item $\rho_\text{AB}$ for given Alice's POVM $\{M^A_{a|x}\}_{a,x}$ and Bob's POVM $\{M^B_{b|y}\}_{b,y}$.
        \end{enumerate}
        Note that each of these optimization problems is an SDP and remains so even if we include the symmetry requirement of~\cref{Eq:Psym}, cf. the first constraint of~\cref{Eq:Opt:SymmCorr}. Moreover, if the optimized strategy turns out to be symmetric, the optimum value would also be a legitimate lower bound of~\cref{Eq:Opt:SQS}. 
        Also, instead of Step~\ref{OptOverB}, one may manually set Bob's POVM to be the same as that of Alice before running step (3) above [while imposing~\cref{Eq:Psym}] to ensure that the resulting strategy is an SQS. However, with this modification, the iterative algorithm is {\em no longer} guaranteed to converge.

\section{Miscellaneous Details}\label{App:MISC}
    \subsection{CGLMP Inequalities and their quantum violation}
        \subsubsection{Unitarity of the operators specified in \cref{Eq:TW}}
        \label{App:Unitarity}

        Here, we provide further details to illustrate the unitarity of the operators given in~\cref{Eq:TW}, and hence $U=TW$.
        In the case of $T$, since it is a diagonal matrix having only eigenvalues $1$ and $-1$, it is unitary.

        To see that $W$ is unitary, it suffices to show that its rows are orthonormal, which from \cref{Eq:TW}, amounts to demanding
        \begin{equation}\label{Eq:Orthonormal}
        	\frac{1}{d^2}\sum _{k=0}^{d-1} \csc \left[\tfrac{ \left(i-k-\frac{1}{2}\right)\pi}{d}\right] \csc \left[\tfrac{ \left(\ell-k-\frac{1}{2}\right)\pi}{d}\right]=\delta_{i,\ell}.
        \end{equation}
        We now give a proof that \cref{Eq:Orthonormal} holds.
        \begin{proof}
        We start by noting the trigonometric identity
        \begin{equation}\label{Eq:Main}
        	\sum_{k=0}^{d-1} \csc\left[x+\frac{k\pi}{d}\right]\csc\left[y+\frac{k\pi}{d}\right] = d\frac{\cot(dx)-\cot(dy)}{\sin(y-x)},
        \end{equation}
        which holds for any real $x,y$ and any integer $d\ge 2$. \cref{Eq:Main} can be easily shown from the more familiar identities:
        \begin{equation}
            \begin{gathered}
            	\csc u \csc v = \frac{\cot u - \cot v}{\sin(v-u)},\quad
            	\sum_{k=0}^{d-1}\cot\left[x+\frac{k\pi}{d}\right]=d\cot(dx),
            \end{gathered}
        \end{equation}
        by setting $u=x+\frac{k\pi}{d}$, $v=y+\frac{k\pi}{d}$, and summing both sides of the first identity over $k$.

        Next, let $x=\frac{\pi}{d}(i+\frac{1}{2}-d)$, $y=\frac{\pi}{d}(\ell+\frac{1}{2}-d)$
        and substitute these into \cref{Eq:Main}. Then, the LHS of \cref{Eq:Main} becomes
        $\sum_{k=0}^{d-1} \csc\left[(i+\frac{1}{2}-d+k)\frac{\pi}{d}\right]\csc\left[(\ell+\frac{1}{2}-d+k)\frac{\pi}{d}\right]$. 
        The value of this expression remains unchanged if we replace the summation index $k$ by $d-1-k$, thus giving, modulo a factor of $d^2$, exactly the LHS of \cref{Eq:Orthonormal}. Incorporating this factor, we may then rewrite the LHS of \cref{Eq:Orthonormal}, via \cref{Eq:Main}, as
        \begin{equation}\label{Eq:Second}
        	\frac{1}{d}\frac{\cot[(i+\frac{1}{2}-d)\pi]-\cot[(\ell+\frac{1}{2}-d)\pi]}{\sin[(\ell-i)\frac{\pi}{d}]}.
        \end{equation}
        
        Consider now the case where $i\neq\ell$, then the denominator of \cref{Eq:Second} never vanishes for $i,\ell\in[d]$. However, since $(i+\frac{1}{2}-d)\pi$ is a half-integer multiple of $\pi$ [likewise for $(\ell+\frac{1}{2}-d)\pi$], their cotangent vanishes, and hence the numerator also vanishes, i.e., \cref{Eq:Second}, and thus the LHS of \cref{Eq:Orthonormal} vanishes whenever $i\neq\ell$. To evaluate \cref{Eq:Second} for the case of $i=\ell$, we apply the l'H\^{o}pital rule to obtain
        \begin{equation}\label{Eq:Third}
        	\csc^2[(i+\frac{1}{2}-d)\pi] =1\quad\forall\,i,\ell\in[d],
        \end{equation}
        since $(i+\frac{1}{2}-d)\pi$ is a half-integer multiple of $\pi$. In other words, the LHS of \cref{Eq:Orthonormal} indeed becomes unity if $i=\ell$. Hence, \cref{Eq:Orthonormal} holds as claimed.
\end{proof}

        \subsubsection{Optimal quantum violation and the minimal dimension $\Dm$}
            \label{App:CGLMP}

            In~\cref{tab:CGLMP-NPA}, we summarize our results concerning the maximal quantum violation of the CGLMP inequality~\cite{CGLMP} $I_d$ for $d=2$ to $19$, obtained via the SQS described in~\cref{Eq:CGLMP-Opt}. Notice that these results generalize those presented in~\cite{ADG+02,NPA2008} for $d=2$ to $8$.
            \begin{table}[H]
                \centering
                \begin{tabular}{|c|c|c|c|c|c|c|c|}
                \hline
                    $d$ & Quantum value & Difference&$d$ & Quantum value & Difference\\ \hline
                     2    & 2.82842718 & $\le$ 5.9e-8& 11     &  3.15549968 & $\le$ 1.9e-7\\
                     3    & 2.91485425 & $\le$ 4.1e-8&12      &  3.16979224 & $\le$ 2.6e-7\\
                     4    & 2.97269840 & $\le$ 1.5e-7&13      &  3.18274300 & $\le$ 2.7e-7\\
                     5    & 3.01571048 & $\le$ 7.7e-9&14      &  3.19456537 & $\le$ 3e-8\\
                     6    & 3.04970041 & $\le$ 1.8e-10&15     &  3.20542659 & $\le$ 1.3e-7\\
                     7    & 3.07764831 & $\le$ 2.4e-9&16      &  3.21546005 & $\le$ 1.34e-6\\
                     8    & 3.10128058 & $\le$ 1.1e-10&17     &  3.22477378 & $\le$ 2.5e-7\\
                     9    & 3.12168442 & $\le$ 1e-8&18        &  3.23345644 & $\le$ 4e-8\\
                    10    & 3.13958741 & $\le$ 1.6e-7&19      &  3.24158164 & $\le$ 9e-8\\
                     \hline
                \end{tabular}
                \caption{Comparison of the quantum violation of the CGLMP inequality $I_d$ based on the SQS of~\cref{Eq:CGLMP-Opt} and the quantum upper bound obtained from NPA Level $1+AB$. Shown in the table is the value of the parameter $d$ from 2 to 19, the quantum value of $I_d$, and its difference from the NPA upper bound, which falls within the numerical precision of the solver.}
                \label{tab:CGLMP-NPA}
            \end{table}

            Next, we show in \cref{tab:CGLMP-Neg} our results for bounding the minimal dimension $\Dm$ required to get the maximal violation (or the maximal-violating correlation) of the CGLMP inequality for $d=2$ to $7$. 

            \begin{table}[!h]
                \centering
                \begin{tabular}{|c|c|c|c|c|c|c|}
                \hline
                    $d$ & Neg ($\ket{\psi_d}$) & Neg($I_d^\star$) & Neg($\vecP^\star$) & Level & Size & $\Dm$  \\ 
                    \hline
                    2 & 0.5    & 0.5000 & 0.5000 & 2  & (25; 31) & 2 \\
                    3 & 0.9836 & 0.9835 & 0.9836 & 2  & (169; 1,003) & 3 \\
                    4 & 1.4561 & 1.4559 & 1.4561 & 2  & (625; 14,797)  & 4 \\
                    5 & 1.9203 & 1.8905 & 1.8909 & 2  & (1,681; 116,702) & 5 \\ 
                    \hline
                    6 & 2.3778 & 1.7578 & 1.7603 & 1+ & (2,025; 118,155) & $\ge$ 5 \\
                    7 & 2.8298 & 1.5896 & 1.5919 & 1+ & (2,809; 97,423) & $\ge$ 5 \\
                    \hline
                \end{tabular}
                \caption{\label{tab:CGLMP-Neg} Negativity and dimension lower bound from the maximal violation of the CGLMP inequality. From left to right, we list the parameter $d$ specifying the CGLMP inequality $I_d$, negativity~\cite{VW02} of the optimal quantum state $\ket{\psi_d}$ obtained from \cref{Eq:CGLMP-Opt}, negativity lower bound based on observing the maximal violation $I_d^\star$ of the CGLMP inequality, negativity lower bound based on observing the maximal $I_d$-violating correlation, the level of the Moroder hierarchy~\cite{Moroder13} involved in the computation of Neg($I_d^\star$) and Neg($\vecP^\star$), the size of the SDP (size of the moment matrix, number of independent moments after symmetrization), and the corresponding lower bound on $\Dm$ for $I_d$ (obtained by comparing against the maximum value of $\frac{d-1}{2}$ attainable for a two-qu\emph{d}it system).}
            \end{table}

            \begin{table*}
            \centering
            \begin{tabular}{|c|c|c|c|c|c|c|c|c|}
                \hline
                \textbf{~Inequality~} & $\Dm$ & \textbf{~Quantum Bound~} & \textbf{~Inequality~} & $\Dm$ & \textbf{~Quantum Bound~} & \textbf{~Inequality~} & $\Dm$ & \textbf{~Quantum Bound~} \\
                \hline \hline
                $A_{5}$  &2& 0.4353  
                & $I_{4422}^{20}$  &4& 0.4677 
                & $J_{4422}^{61}$  &2& 0.8175
                \\
                
                $A_{6}$  &4& 0.3004 
                & $J_{4422}^{12}$  &4& 0.7262 
                & $J_{4422}^{85}$  &2& 0.9763
                \\
                
                $A_{7}\,\bigl(I_{4422}^1\bigr)$ &3& 0.2879 
                & $J_{4422}^{17}$  &2&0.6380 
                & $J_{4422}^{90}$  &2&0.8398
                \\
                
                $AII_{1}$  &2& 0.6056 
                & $J_{4422}^{18}$  &2& 1.0130  
                & $J_{4422}^{91}$  &2& 1.2993
                \\
                
                $AS_{1}$  &2& 0.5413 
                & $J_{4422}^{19}$  &2& 0.6742 
                & $J_{4422}^{92}$  &2& 1.0648
                \\
                
                $AS_{2}$  &2& 0.8785 
                & $J_{4422}^{22}$  &2& 0.8156 
                & $J_{4422}^{97}$  &2& 1.1584
                \\
                
                $I_{4422}^9$  &2& 0.4617 
                & $J_{4422}^{26}$  &2& 0.6402 
                & $J_{4422}^{102}$ &2&0.6651
                \\
                
                $I_{4422}^{10}$ &2& 0.6139 
                & $J_{4422}^{27}$  &2& 0.8972 
                & $J_{4422}^{105}$ &2& 1.0742
                \\
                
                $I_{4422}^{11}$  &2& 0.6383 
                & $J_{4422}^{28}$  &2& 0.7500 
                & $J_{4422}^{108}$ &2& 0.9677
                \\
                
                $I_{4422}^{12}$  &2& 0.6188 
                & $J_{4422}^{29}$  &2& 1.0246 
                & $J_{4422}^{109}$ &2& 1.7261
                \\
                
                $I_{4422}^{14}$  &2& 0.4794 
                & $J_{4422}^{32}$  &2& 0.5901 
                & $J_{4422}^{110}$ &2& 0.9457
                \\
                
                $I_{4422}^{16}$  &2& 0.4142 
                & $J_{4422}^{41}$  &2& 0.7596 
                & $J_{4422}^{125}$ &2& 1.0000
                \\
                
                $I_{4422}^{17}$  &2& 0.6714 
                & $J_{4422}^{58}$  &2& 0.8814 
                & $S_{242}^{51}$  &2& 1.0135 
                \\
                
                $I_{4422}^{18}$  &3& 0.6430 
                & $J_{4422}^{60}$  &2& 0.5923 
                & $S_{242}^{52}$  &2& 0.8704\\
                
                \hline
            \end{tabular}
            \caption{\label{tbl:symmetry class} List of {\em all} symmetric {\em facet-defining} four-setting, two-outcome Bell inequalities whose maximal quantum violation (denoted as ``Quantum Bound'') is achievable using an SQS of minimal dimension $\Dm$. In other words, to attain the maximal quantum violation of these inequalities, there is no trade-off between symmetry and dimension.
            }
            \end{table*}
        
    \subsection{SQS maximally violating an $I_{3322}$-like inequality}
        \label{App:I3322LikeSQS}
        Consider the Bell inequality from~\citep[Eq. (27)]{Goh2018}
           \begin{align}\label{Eq:CorI3322}
                   I_{3322c} = &\expA{A_0B_1}+\expA{A_0B_2}+\expA{A_1B_0}+\expA{A_2B_0}\\
                    +& \expA{A_0B_0}+\expA{A_1B_1}-\expA{A_1B_2} - \expA{A_2B_1}\overset{\L}{\le} 4,\nonumber
           \end{align}
        which can be seen, after relabeling, as keeping only the correlation part of the $I_{3322}$ inequality. The optimal QS presented in~\cite{Goh2018} is {\em not} PPI. However, we can turn it into an SQS by first applying $\sigma_z$, followed by $R_{\hat{x}}(\alpha)$ on Bob's Hilbert space.
        After this rotation, the strategy becomes
        \begin{equation}\label{Eq:I3322cStrategy}
           \begin{aligned}
                \ket{\psi} &=i\sin\frac{\alpha}{2}\ket{\Phi^+}-\cos\frac{\alpha}{2}\ket{\Psi^+},\\
               \hat{A}_0 &= \hat{B}_0 = \frac{1}{2}\big( 2 \cos{\frac{\pi}{6}}\, \sigma_x + \cos{\alpha}\, \sigma_y + \sin{\alpha}\, \sigma_z \big),\\
               \hat{A}_1 &= \hat{B}_1 = \frac{1}{2}\big( 2 \cos{\frac{\pi}{6}}\, \sigma_x - \cos{\alpha}\, \sigma_y - \sin{\alpha}\, \sigma_z \big),\\
               \hat{A}_2 &= \hat{B}_2 = \sigma_y,
           \end{aligned}
        \end{equation}
        which is easily verified to give a quantum value of $5$ for all $\alpha \in [0,2\pi]$. Notice that, by construction, despite its dependence on $\alpha$, the two-qubit state of \cref{Eq:I3322cStrategy} is always maximally entangled.
    \subsection{Optimal quantum strategy for \texorpdfstring{$I_{S}(\alpha)$}{IS}}
        \label{App:Is}
    
    An {\em asymmetric} two-qubit QS that gives the maximal violation of the family of Bell inequalities of~\cref{eq:I_S} consists of Alice and Bob sharing the two-qubit state
    \begin{subequations}\label{eq:OptimalQS:Is}
        \begin{equation}\label{eq:psip}
            \ket{\psi} = \sqrt{p}\ket{00}+\sqrt{1-p}\ket{11},   
        \end{equation}
        and measuring the dichotomic observables:
        \begin{equation}\label{eq:obs-Is}
            \begin{gathered}
                \hat{A}_0 = s\,\sigma_x + \sqrt{1-s^2}\sigma_z,\quad
                \hat{A}_1 = \sigma_x,\quad \hat{A}_2 = \sigma_z,\\
                \hat{B}_{y} = t\,\sigma_x +(-1)^y \sqrt{1-t^2}\sigma_z\,\,\text{for } y=0,1, \text{ and }\,\,
                \hat{B}_2 = \mathbb{I}_2.
            \end{gathered}
        \end{equation}
        For $\alpha = 1.5$, the actual values of these parameters are
        \begin{equation}\label{eq:para_ab:1.5}
            p =\tfrac{16}{21},\quad s = \tfrac{8}{17},\quad
            t = -\tfrac{\sqrt{5}}{3},
        \end{equation}
        thereby giving $I_S(1.5) = 8\frac{1}{3}$, while those for $\alpha = 2$ are
        \begin{equation}\label{eq:para_ab:2}
            p =\tfrac{2(1+\sqrt{13})}{3\sqrt{13}},\quad 	s = \tfrac{10\sqrt{13}-18}{61},\quad
            t = -\sqrt{\tfrac{11-\sqrt{13}}{18}}.
        \end{equation}
    \end{subequations}

    \subsection{Symmetric bipartite facet-defining Bell inequalities with four binary-outcome measurements}
        \label{App:4422}

    Among the complete list of 175 facet-defining Bell inequalities of the Bell scenario $(2,4,2)$, 55 of them can be cast in a symmetric form.
    Apart from the positivity facet, the CHSH Bell inequality of~\cref{Eq:CHSH}, the (symmetric) $I_{3322}$ Bell inequality~\cite{Sliwa03,Brunner2008}, the 9 inequalities listed in~\cref{tbl: asymmetric strategy}, and the $J^{42}_{4422}$ inequality,
    the remaining 42 symmetric four-setting facet-defining Bell inequalities, listed in~\cref{tbl:symmetry class}, show no trade-off between symmetry and dimension.
    
    For completeness, we provide here the explicit form of the $J^{42}_{4422}$ Bell inequality:
    \begin{align}\label{Ineq:J42}
        J^{42}_{4422} = 
        &\sum_{i ={\A,\B}}- \left[P_i(0|0) +P_i(0|3) +4P_i(0|1)\right]\nonumber\\
        &+\sum_{x = 0}^2\PAB(0,0|x,\text{mod}(2(x+1),3))\\
        &+  \sum_{x\neq 1} 2\left[\PAB(0,0|1,x) + \PAB(0,0|x,1)\right]\nonumber\\
        &+  \sum_{x= 0,2}\left[ \PAB(0,0|3,x) + \PAB(0,0|x,3)\right]\nonumber\\
        &-2\sum_{x=2}^3 \PAB(0,0|x,x)-3\PAB(0,0|0,0)
         \, \overset{\L}{\le} 0. \nonumber
    \end{align}

	For the various SQS (and symmetric-correlation) bounds of those inequalities from~\cref{tbl:symmetry class} with $\Dm>2$, see~\cref{tbl:HighD-symmetryclass}.
	
        \begin{table}[h!]
        \centering
            \begin{tabular}{|c|c|c|c|c|c|c|c|}
                \hline
                \textbf{~Inequality~} &$\Dm$&$d = 2$ &$d = 3$& \textbf{~Quantum Bound~} \\
                \hline \hline
                $A_{6}$  &4&0.2990&0.2990& 0.3004 \\
                $A_{7}\,\bigl(I_{4422}^1\bigr)$ &3&0.2500&0.2879 &0.2879 \\
                $I_{4422}^{18}$  &3&0.4676&0.6430 & 0.6430 \\
                $I_{4422}^{20}$  &4&0.3056&0.3662 (0.4362) & 0.4677 \\
                $J_{4422}^{12}$ &4&0.6719&0.6830 & 0.7262 \\
                \hline
            \end{tabular}
        \caption{\label{tbl:HighD-symmetryclass} Subset of inequalities from~\cref{tbl:symmetry class} whose minimal dimension $\Dm>2$. The third and fourth columns show, respectively, the SQS bound for $d=2$ and $d=3$, which also coincides with the symmetric correlation bound derived from~\cite{TRR19}, except for $I_{4422}^{20}$, where the latter bound is included in parentheses. 
        }
        \end{table}
        
    \subsection{Quantum violation of $I_9$}
    \label{App:I9}
    
    In this Appendix, we shall present a symmetric Bell inequality $I_9$ defined for the $(2,9,2)$ Bell scenario such that there is again a trade-off between symmetry and dimension when one maximizes its Bell violation. In particular, even though the maximal quantum violation of $I_9$ can be achieved using a symmetric correlation arising from a symmetric two-qubit state and a mirror-symmetric measurement strategy, the maximal quantum value attainable using a two-qubit SQS is always suboptimal. 
    
    Let us first recall the self-testing results shown in~\cite{Bowles18}.
    There, the authors considered the following Bell inequality (introduced in~\cite{APVW16}):
    \begin{equation}
        \begin{aligned}
            \text{CHSH}_3 
            &= \text{CHSH}(1,2;4,5) + \text{CHSH}(1,3;6,7)\\
            &+ \text{CHSH}(2,3;8,9),
        \end{aligned}
    \end{equation}
    where, see~\cref{Eq:CHSH2},  
    \begin{equation}
        \text{CHSH}(x_1,x_2;y_1,y_2) \coloneqq \braket{A_{x_1} (B_{y_1}+B_{y_2})} + \braket{A_{x_2}  (B_{y_1}-B_{y_2})}.   
    \end{equation}
    
    The maximal quantum violation of $\text{CHSH}_3$ is easily seen to be $3\times2\sqrt{2}=6\sqrt{2}$, which one can attain by locally measuring the Bell state $\ket{\Phi^+}$ with the following dichotomic observables:
    \begin{equation}\label{Eq:A-B}
        \begin{gathered}
            \hat{A}_1 = \sigma_z,\quad \hat{A}_2 = \sigma_x,\quad \hat{A}_3 = -\sigma_y,\\
            \hat{B}_4 = \frac{\sigma_z + \sigma_x}{\sqrt{2}}, 
            \quad \hat{B}_5 = \frac{\sigma_z - \sigma_x}{\sqrt{2}},\\
            \hat{B}_6 = \frac{\sigma_z + \sigma_y}{\sqrt{2}}, 
            \quad \hat{B}_7 = \frac{\sigma_z - \sigma_y}{\sqrt{2}},\\
            \hat{B}_8 = \frac{\sigma_x + \sigma_y}{\sqrt{2}}, 
            \quad \hat{B}_9 = \frac{\sigma_x - \sigma_y}{\sqrt{2}}.
        \end{gathered}
    \end{equation}
    In fact, it was shown in~\cite{Bowles18} that the maximal quantum value of $6\sqrt{2}$ self-tests the two-qubit maximally entangled state $\ket{\Phi^+}$ and, up to complex conjugation, the Pauli observables $\hat{A}_x$ for $x=1,2,$ and $3$.
    Now consider the following Bell inequality in the scenario of $(2,9,2)$:
    \begin{equation}
        I_9 = \text{CHSH}_3 + \text{CHSH}_3' + \sum_{k=1}^3 \braket{A_k  B_k},
    \end{equation}
    where $\text{CHSH}_3'$ is obtained from $\text{CHSH}_3$ by swapping the roles of Alice and Bob, i.e., $\text{CHSH}_3' = \text{CHSH}(4,5;1,2) + \text{CHSH}(6,7;1,3) + \text{CHSH}(8,9;2,3)$.
    The inequality $I_9$ is symmetric according to~\cref{def:Symmetric Bell inequality}.

    \begin{proposition}
        The maximal quantum value of $I_9$ cannot be achieved by any symmetric qubit quantum strategy.
    \end{proposition}
    \begin{proof}
        The maximal quantum value of $I_9$, namely $12\sqrt{2} + 3$, can be attained by locally measuring $\ket{\Phi^+}$ with the observables $\{\hat{A}_x\}_{x=1}^3$, $\{\hat{B}_y\}_{y=4}^9$ specified in \cref{Eq:A-B} and their complex conjugation, i.e., $\hat{A}_k:=(\hat{B}_k)^*$ for $k=4$ to 9 and $\hat{B}_k:=(\hat{A}_k)^*$ for $k=1$ to 3.
        In particular, any realization achieving the maximal value of $I_9$ also achieves the maximal value of $\mathrm{CHSH}_3$.
        According to the self-testing result of~\cite{Bowles18}, the underlying state $\ket{\psi}$ is therefore locally isometric to a two-qubit maximally entangled state, $\ket{\Phi^+}$.
        In particular, the state $\ket{\psi}$ in any qubit strategy reaching the maximal value of $I_9$ can be written as $\ket{\psi} = \I \ten U \ket{\Phi^+}$ for some $U \in SU(2)$.
        The identity $V \ten \I \ket{\Phi^+} = \I \ten V^T \ket{\Phi^+}$ allows us to consider local unitaries operating on only one side without loss of generality.
        
        Suppose, for contradiction, that there exists some symmetric qubit strategy achieving the maximal value of $I_9$, i.e., $A_i = B_i$, for $i=1,2,3$.
        To reach $12\sqrt{2}+3$, the three diagonal correlators $\braket{A_i  B_i}$ must all be $1$.
        Then for each $i\in\{1,2,3\}$, we have
        \begin{equation}
            1 
            = \braket{\psi   | \hat{A}_i \ten \hat{A}_i | \psi} 
            = \braket{\Phi^+ | \hat{A}_i \ten U^\dagger \hat{A}_i U | \Phi^+} 
            \leq 1,
        \end{equation}
        where the last inequality follows from the Cauchy-Schwarz inequality. 
        Since the inequality is saturated, we must have
        $
        {\ket{\Phi^+} 
        = \hat{A}_i \ten U^\dagger \hat{A}_i U \ket{\Phi^+}}$ $=\id_2 \ten (\hat{A}_i)^\text{\tiny T} U^\dagger \hat{A}_i U \ket{\Phi^+}
        $, where $(.)^\text{\tiny T}$ denotes transposition.
        Applying $\bra{0} \ten \id_2$ and $\bra{1} \ten \id_2$ on both sides, we infer that $(\hat{A}_i)^\text{\tiny T} U^\dagger \hat{A}_i U = \id_2$, which further implies that
        \begin{equation}\label{Eq:A_T}
            (\hat{A}_i)^\text{\tiny T} = U^\dagger \hat{A}_i U\quad\text{for } i=1,2,3,
         \end{equation}
         where we have used the fact a qubit observable $\hat A_i$ [and hence $(\hat{A}_i)^\text{\tiny T}$] squares to the identity operator.
        Equivalently, \cref{Eq:A_T} means that there exists some rotation $\mathcal{R} \in SO(3)$ satisfying
        \begin{equation}\label{Eq:mirror-rot}
            (\mathcal{M} \vec{a}_i ) \cdot \vec{\sigma} = (\mathcal{R} \vec{a}_i) \cdot \vec{\sigma},
        \end{equation}
        where $\hat{A}_i=\vec{a}_i\cdot\vec\sigma$ and $\mathcal{M}=\text{diag}(1,-1,1)$.
        Heuristically, it asserts the existence of some rotation $\mathcal{R}$ that can simultaneously map all $\vec{a}_i$s to each of their own reflection $\vec{a}_i' = \mathcal{M} \vec{a}_i, i = 1,2,3$.
        
        However, at the maximal $\mathrm{CHSH}_3$ value, the observables $A_1,A_2,A_3$ correspond to three mutually orthogonal Bloch vectors~\cite{Bowles18} $\vec{a}_1, \vec{a}_2, \vec{a}_3$, which cannot be coplanar.
        Hence, as explained in Proposition~\ref{Prop:Nonplanar}, there cannot be a rotation $\mathcal{R}$ that satisfies \cref{Eq:mirror-rot}, thus leading to a contradiction.
        Therefore, no symmetric qubit strategies could reach the maximal quantum value of $I_9$.
    \end{proof}

    What then is the largest value of $I_9$ achievable with a qubit SQS? To this end, the best violation that we have found using a qubit SQS is $6\sqrt{3} +9\simeq 19.3923 < 12\sqrt{2} + 3$, where the SQS consists of measuring 
    the Bell state $\ket{\Phi^+}$ with the observables
    \begin{equation}
        \begin{gathered}
            \hat{A}_1 = \hat{A}_9 = \sigma_z,\quad \hat{A}_8 = \sigma_x,\\
            \hat{A}_2 = \hat{A}_6 = \frac{\sqrt{3}\sigma_x +\sigma_z}{2},\\
            \hat{A}_3 = -\hat{A}_5 = \frac{\sqrt{3}\sigma_x -\sigma_z}{2},\\
            \hat{A}_4 = \frac{\sigma_x +\sqrt{3}\sigma_z}{2},\quad
            \hat{A}_7 = \frac{-\sigma_x +\sqrt{3}\sigma_z}{2},
        \end{gathered}
    \end{equation}
    and $\hat{B}_i = \hat{A}_i$ for $i = 1 $ to $9$. Note that all these observables are real, and thus give rise to measurement (Bloch) vectors $\vec{a}_k =\vec{b}_k$ lying on the $x-z$ plane.

    \subsection{Local bound for the Bell inequalities of \cref{Eq:TsirelsonIneqInspired}}
    \label{App:LocalBound}
    
    The local upper bound $g(r_0,r_1)$ for the Bell expression defined in~\cref{Eq:TsirelsonIneqInspired} is, for any given pair $(r_0,r_1)$:
    \begin{equation}\label{Eq:g}
    	\max \left\{ \frac{1}{\sqrt{2}} \pm r_0 \pm (\sqrt{2}-1) r_1, \frac{1}{\sqrt{2}} \pm r_1 \pm (\sqrt{2}-1) r_0 \right\},
    \end{equation}
    where the $\pm$ expression in each term allows {\em all} combinations of signs. The actual bound depends on the octagon slice [spanned by $(0,0)$ and the vertices of \cref{Eq:Vertices}] to which the point $(r_0,r_1)$ belongs. 
    
    For instance, consider the case where $(r_0,r_1)$ belongs to the octagon slice spanned by $(0,0)$, $(\frac{1}{2}-\zeta,0)$, and $(\zeta,\zeta)$, with $\zeta:= \frac{1}{\sqrt{2}}-\frac{1}{2}$.
    Since $r_0\ge r_1\ge 0$, evaluating the maximum from \cref{Eq:g} gives 
    \begin{equation}
    	g(r_0,r_1) = \frac{1}{\sqrt{2}} + r_0 + (\sqrt{2}-1) r_1. 
    \end{equation}
    Similarly, if $(r_0,r_1)$ belongs to the octagon slice spanned by $(0,0)$, $(0,\zeta-\frac{1}{2})$, and $(-\zeta,-\zeta)$, then $r_1\le r_0\le 0$, thus 
    evaluating the maximum from \cref{Eq:g} gives 
    \begin{equation}
    	g(r_0,r_1) = \frac{1}{\sqrt{2}} - r_1 - (\sqrt{2}-1) r_0. 
    \end{equation}
    The local upper bound $g(r_0,r_1)$ of $I_{r_0,r_1}$ for $(r_0,r_1)$ belonging to the other six octagon slices can be easily deduced accordingly.

\section{Multipartite Considerations}
\label{App:Multipartite}    

\subsection{Purified, PPI multipartite quantum strategy for PPI correlation}

Here, we provide the PSQS alluded to in Proposition~\ref{Prop:SCor2SQS}. To this end, let $\Qstr'$ be an $N$-partite QS realizing a PPI $N$-partite correlation $\vecP$:
    \begin{equation}\label{Eq:MPQS}
	\begin{gathered}
    	\Qstr'=\left\{\ket{\psi}_{\A_1\cdots\A_N}, \{\Pi_{a_1|x_1}^{\A_1}\}_{a_1,x_1},\cdots,\{\Pi_{a_N|x_N}^{\A_N}\}_{a_N,x_N}\right\},\\
		(\Pi_{a_i|x_i}^{\A_i})^2=\Pi_{a_i|x_i}^{\A_i}\,\,\forall\,\,a_i,x_i,i.
	\end{gathered}	
    \end{equation} 
    
 Furthermore, let $S_N$ be the permutation group of $N$ objects and $\sigma$ be its elements.
Then, the analog of \cref{Eq:PSQS,eq:pure-symmetrization} for the $N$-partite scenario may be constructed as:
    \begin{equation}\label{Eq:NPSQS}
    	\tilde{\Qstr}:=\{\tilde{\ket{\phi}}_{\A_1\cdots \A_N}, \{\tilde{\Pi}_{a_1|x_1}\}_{a_1,x_1}, \cdots,\{\tilde{\Pi}_{a_N|x_N}\}_{a_N,x_N}\},
	\end{equation}
    where the symmetrized state is 
    \begin{equation}\label{eq:pure-symmetrization-Npartite}
    \begin{aligned}
        \tilde{\ket{\phi}}_{\A_1\cdots \A_N\A_1'\cdots \A_N'}     =
        \frac{1}{\sqrt{N!}}
        \sum_{\sigma\in S_N}
        &U_{\sigma}\ket{\psi}_{\A_1\cdots \A_N}\\
        \otimes &V_{\sigma}\ket{0, \cdots, N-1}_{\A'_{1}\cdots \A'_{N}},
    \end{aligned}
\end{equation}
and the local measurement projectors are
\begin{equation}
    \begin{aligned}
        \tilde{\Pi}_{a|x}
        = \sum_{i=0}^{N-1}\Pi^{A_{i+1}}_{a|x}\otimes\ketbra{i}{i}.
    \end{aligned}
\end{equation}
Here, $U_{\sigma}$ and $V_{\sigma}$ are unitary representations of $\sigma$ that, respectively, permute the $N$ Hilbert spaces of $\A_1, \A_2,\cdots, \A_N$ and $\A'_1, \A'_2,\cdots, \A'_N$. Moreover, the Hilbert spaces for $\A'_1$, $\A'_2$, $\cdots$, $\A'_N$ are all isomorphic to $\mathbb{C}^N$ and spanned by the orthonormal basis $\{\ket{i}\}_{i=0}^{N-1}$. Hence, if $\ket{\psi}$ is a vector of the Hilbert space $\mathbb{C}^{d^N}$, then $\tilde{\ket{\phi}}$ is an element of the Hilbert space $\mathbb{C}^{(dN)^N}$. Note that a similar construction to obtain a (mixed) QS that is translational-invariant has previously been presented in~\cite{Tura_2014}.

\subsubsection{The tripartite case}

For $N = 3$, the permutation group $S_3$ consists of six elements. Thus, the symmetrized state consists of an equal-weight superposition of six permuted terms:
\begin{equation}\label{eq:pure-symmetrization-tripartite}
    \begin{aligned}
        &\tilde{\ket{\phi}}_{\A\B\tC\A'\B'\tC'}
        =\\
        \frac{1}{\sqrt{6}}
        \Big[
        &\ket{\psi}_\text{ABC}\ket{012}_{\A'\B'\tC'}
        +\ket{\psi}_\text{ACB}\ket{021}_{\A'\B'\tC'}\\
        +&\ket{\psi}_\text{BAC}\ket{102}_{\A'\B'\tC'} 
        +\ket{\psi}_\text{BCA}\ket{120}_{\A'\B'\tC'}\\
        +&\ket{\psi}_\text{CAB}\ket{201}_{\A'\B'\tC'}
        +\ket{\psi}_\text{CBA}\ket{210}_{\A'\B'\tC'}
        \Big],
    \end{aligned}
\end{equation}
and the local measurement projectors are
\begin{equation}\label{eq:sym-tripartite-POVM}
    \begin{aligned}
        \tilde{\Pi}_{a|x}
        =
        \Pi^{\A}_{a|x}\otimes\ketbra{0}{0}
        +\Pi^{\B}_{a|x}\otimes\ketbra{1}{1}
        +\Pi^{\tC}_{a|x}\otimes\ketbra{2}{2}.
    \end{aligned}
\end{equation}

We may use \cref{eq:pure-symmetrization-tripartite,eq:sym-tripartite-POVM} to symmetrize, e.g., the three-qutrit asymmetric strategy presented in~\cite{BBGL11} for maximizing the quantum violation of $S_{3,3}$ discussed therein. However, a closer inspection shows that there is again no trade-off for this Bell-type inequality, i.e., its maximal quantum violation can already be attained by considering a PPI three-qutrit strategy.

\subsection{Examples of multipartite Bell inequalities where SQS in the minimal dimension can be maximizing}

For completeness, we list here several other multipartite (facet-defining) Bell inequalities that provably do not exhibit a trade-off between symmetry and dimension for maximizing their Bell violation.

\subsubsection{Inequalities in the $(3,2,2)$ Bell scenario}

The complete set of 46 facet-defining Bell inequalities for this scenario was first presented by Sliwa in~\cite{Sliwa03}. Apart from the positivity facet (inequality 1 in the list), we have found that only inequalities 2, 5, 7, 22, 26, 33, and 39 can be cast in a PPI form (after an appropriate relabeling). However, for all these inequalities cast in their PPI form, we have found a PPI three-qubit QS~\cite{LRXC16} that matches their maximal quantum violation (see also~\cite{Vallins2017}).

\subsubsection{Inequalities in the $(N,2,2)$ Bell scenario}

In~\cite{Liang15}, the 7th inequality presented by Sliwa~\cite{Sliwa03} for the $(3,2,2)$ Bell scenario has been generalized to the Bell scenarios $(N,2,2)$ for an arbitrary integer $N\ge 3$. Moreover, an explicit PPI $N$-qubit strategy matching the quantum maximum has also been presented.

\end{document}